\documentclass[11pt]{article}

\usepackage{amsmath}
\usepackage{enumitem}
\usepackage{amsfonts}
\usepackage{amsthm}
\usepackage{amssymb}

\usepackage{textcomp}

\usepackage[margin=1in]{geometry}
\usepackage{indentfirst}

\usepackage{abstract}

\usepackage{pbox}
\usepackage{comment}
\usepackage{pifont}

\usepackage[boxruled,noend]{algorithm2e}
\usepackage{pgfplots}
\pgfplotsset{compat=1.4}
\usetikzlibrary{shapes}
\usetikzlibrary{plotmarks}

\usepackage{framed}
\usepackage{mdframed}
\usepackage{placeins}
\usepackage{subcaption}
\usepackage{afterpage}
\usepackage[colorlinks,citecolor=blue,linkcolor=blue,urlcolor=red,pagebackref]{hyperref}
\usepackage{courier}
\usepackage{multirow}
\usepackage{mathrsfs}
\usepackage{microtype}

\usepackage{stmaryrd}
\usepackage[capitalise]{cleveref}
\usepackage{thm-restate}

\allowdisplaybreaks

\usepackage[numbers]{natbib}

\makeatletter
\DeclareRobustCommand\citep
{\begingroup\tiny\NAT@swatrue\let\NAT@ctype\z@\NAT@partrue
    \@ifstar{\NAT@fulltrue\NAT@citetp}{\NAT@fullfalse\NAT@citetp}}
\makeatother

\newcommand{\eps}{\varepsilon}
\newcommand{\Oish}{\widetilde{O}}

\newcommand{\outdeg}{\text{\normalfont outdeg}}

\DeclareMathOperator{\dist}{dist}

\newcommand{\polylog}{\operatorname{polylog}}

\newcommand{\hopdist}{\texttt{hopdist}}
\newcommand{\tO}{\widetilde{O}}

\newcommand{\UDP}{\textsc{UDP}}
\newcommand{\DDP}{\textsc{DDP}}
\newcommand{\DADP}{\textsc{DADP}}

\newtheorem{theorem}{Theorem}[section]
\newtheorem{lemma}[theorem]{Lemma}

\newtheorem{corollary}[theorem]{Corollary}
\newtheorem{observation}[theorem]{Observation}

\newtheorem{claim}[theorem]{Claim}
\newtheorem{open}{Open Question}

\theoremstyle{definition}
\newtheorem{definition}[theorem]{Definition}

\usepackage{amsmath}

\newtheorem*{definition*}{Definition}
\newtheorem*{theorem*}{Theorem}

\newtheorem*{lemma*}{Lemma}

\newcommand{\ints}[1]{\llbracket #1 \rrbracket}

\title{New Separations and Reductions for Directed Preservers and Hopsets}
\author{
Gary Hoppenworth\thanks{University of Michigan. \texttt{garytho@umich.edu}. Supported by NSF:AF 2153680. } \and Yinzhan Xu\thanks{Massachusetts Institute of Technology. \texttt{xyzhan@mit.edu}. Supported by NSF Grant CCF-2330048, a Simons
Investigator Award and HDR TRIPODS Phase II grant 2217058. } \and Zixuan Xu\thanks{Massachusetts Institute of Technology. \texttt{zixuanxu@mit.edu}.}
}
\date{}

\begin{document}

\maketitle

\begin{abstract}
    We study distance preservers, hopsets, and shortcut sets in $n$-node, $m$-edge directed graphs, and show improved bounds and new reductions for various settings of these problems. 
    Our first set of results is about exact and approximate distance preservers. We give the following bounds on the size of directed distance preservers with $p$ demand pairs:
    \begin{itemize}
        \item $\Oish(n^{5/6}p^{2/3} + n)$ edges for exact distance preservers in unweighted graphs, and
        \item $\Omega(n^{2/3}p^{2/3})$ edges for approximate distance preservers with any given finite stretch, in graphs with arbitrary aspect ratio.
    \end{itemize}

Additionally, we establish a new directed-to-undirected reduction for exact distance preservers. We show that if undirected distance preservers have size $O(n^{\lambda}p^{\mu} + n)$ for constants $\lambda, \mu > 0$, then directed distance preservers have size $O\left( n^{\frac{1}{2-\lambda}}p^{\frac{1+\mu-\lambda}{2-\lambda}} + n^{1/2}p + n\right).$ As a consequence of the reduction, if state-of-the-art upper bounds for undirected preservers can be improved for some $p \leq n$, then so can  state-of-the-art upper bounds for directed preservers.

Our second set of results is about directed hopsets and shortcut sets. For hopsets in directed graphs, we prove that the hopbound is:
\begin{itemize}
    \item $\Omega(n^{2/9})$ for $O(m)$-size shortcut sets,  improving the previous $\Omega(n^{1/5})$ bound [Vassilevska Williams, Xu and Xu, SODA 2024],
    \item $\Omega(n^{2/7})$ for $O(m)$-size exact hopsets in unweighted graphs, improving the previous $\Omega(n^{1/4})$ bound [Bodwin and Hoppenworth, FOCS 2023], and
    \item $\Omega(n^{1/2})$ for $O(n)$-size approximate hopsets with any given finite stretch, in graphs with arbitrary aspect ratio. This result establishes a separation between this setting and $O(n)$-size approximate hopsets for graphs with polynomial aspect ratio,  which are known to have an $\widetilde{O}(n^{1/3})$ upper bound [Bernstein and Wein, SODA 2023].
\end{itemize}
\end{abstract}

\setcounter{page}{0}
\thispagestyle{empty}

\newpage

\tableofcontents
\setcounter{page}{0}
\thispagestyle{empty}
\newpage

\section{Introduction}

In network design, two commonly studied categories of problems are graph augmentation problems, where the goal is to augment the graph so that it gains some useful property, and graph sparsification problems, where the goal is to sparsify the graph while preserving some of its features. In this work, we study two fundamental problems from these two categories: hopsets and distance preservers. 

The hopset problem  is a typical graph augmentation problem. We are given a  graph,\footnote{We assume graphs to be directed unless otherwise specified.}
and we want to add a small number of additional edges to the graph so that all pairs of vertices $(s, t)$ have an approximate shortest $s \leadsto t$ path with few edges. Another popular graph augmentation object, the shortcut set (first introduced by \cite{Thorup92}), is actually a special case of hopsets where we only care about reachability. Hopsets and shortcut sets have a wide range of applications in areas such as parallel computing and distributed computing (e.g.~\cite{KS97,HKN14b,EN16,FN18,JLS19,KS21,EN19,Fineman19,ASZ20,BGW20,CFR20,ElkinM21}).

On the other hand, the (approximate) distance preserver problem is a graph sparsification problem. Here, we are given an input graph $G$ and a set of demand pairs $P \subseteq V(G) \times V(G)$, and we need to find a sparse subgraph $H$ of $G$ that preserves all distances between demand pairs in $G$ up to an approximation factor $\alpha$. Approximate distance preservers are a generalization of another well-studied graph sparsification object,  reachability preservers. Distance preservers have applications to problems such as distance oracles \cite{EP16, DBLP:conf/focs/ElkinS23} and spanners~\cite{BCE05, BV15, AB16soda, AB17jacm, ABP17, EFN17, GK17, BV21}.

Despite being graph augmentation and graph sparsification problems respectively, hopsets and distance preservers are known to be related in some ways. For instance, distance preserver lower bounds imply lower bounds for hopsets \cite{KP22, BHT22}, and  small-stretch hopsets can be translated into  small-stretch distance preservers in a black-box way \cite{KP22}.

In this paper, we develop better understandings of hopsets and distance preservers in directed graphs. We  discuss our results and compare them with prior work in \cref{sec:intro:hopset,sec:intro:preserver}.

\subsection{ Preservers}
\label{sec:intro:preserver}

The problem of constructing sparse distance preservers is an important basic combinatorial problem with connections to incidence geometry \cite{CE06, kalia2016generalizations}, distance oracles \cite{EP16, DBLP:conf/focs/ElkinS23}, and fast graph algorithms \cite{Alon02, CGMW18, DBLP:conf/focs/ElkinS23}. We now formally define the more general object of approximate distance preservers, or approximate preservers for short. 

\begin{definition}[Approximate Preservers]
    Given an $n$-node graph $G$ and a set of demand pairs $P \subseteq V(G) \times V(G)$ of size $|P| = p$ in the transitive closure of $G$, a subgraph $H$ of $G$ is an $\alpha$-approximate distance preserver of $G, P$ if for all $(s, t) \in P$, $s$ can reach $t$ in $H$, and 
    $$ \dist_H(s, t) \leq \alpha \cdot \dist_G(s, t).$$ 
\end{definition}
\noindent There are three settings of approximate preservers that have received special attention in prior work:
\begin{itemize}
    \item $\alpha = 1$ \qquad \hspace{3mm} (also known as \textit{Exact Distance Preservers}.)
    \item $\alpha = 1 +  \varepsilon$ \quad \hspace{0.5mm} (also known as $(1+\varepsilon)$-\textit{Approximate Preservers}.)
    \item $\alpha = \infty$ \qquad \hspace{1mm} (also known as \textit{Reachability Preservers.})\footnote{Naturally, reachability preservers only preserve the reachability between demand pairs.}
\end{itemize}
\newpage

\begin{table}[ht]
\small
    \centering
    \begin{tabular}{ c c c c}
    \hline 
        Setting  & Upper Bound & Lower Bound & Beats Consistency? \\
        \hline 
        \hline 
\begin{tabular}{@{}c@{}} $\alpha = 1$  \\ undirected 
\end{tabular}
        & 
        \begin{tabular}{@{}c@{}} $O(\min\{n^{1/2}p + n, np^{1/2}\})$  \\ \citep{CE06} \end{tabular}
          & 
                  \begin{tabular}{@{}c@{}} $\Omega(n^{2/3} p^{2/3} + n)$  \\ \citep{CE06} \end{tabular}
          & \ding{55}          \\
                  \hline 
\begin{tabular}{@{}c@{}} $\alpha = 1$  \\ undirected \\
unweighted \end{tabular}
        & 
        \begin{tabular}{@{}c@{}} $O\left(\min\left\{n^{2/3}p^{2/3}+np^{1/3},  \frac{n^2}{\texttt{RS}(n)}\right\}\right)$  \\ \citep{BV21, bodwin2017linear} \end{tabular}
          & 
                  \begin{tabular}{@{}c@{}}$\Omega(n^{2/(d+1)}p^{(d-1)/d})$ \\ for integer $d \ge 1$\quad \citep{CE06}\end{tabular}
          & \ding{51}          \\
        \hline
        \begin{tabular}{@{}c@{}} $\alpha = 1$  \\ directed \end{tabular} &
        \begin{tabular}{@{}c@{}} $O(\min\{n^{2/3}p + n, np^{1/2}\})$  \\ \citep{CE06, Bodwin21} \end{tabular} &
        \begin{tabular}{@{}c@{}} $\Omega(n^{2/3} p^{2/3} + n)$  \\ \citep{CE06} \end{tabular}  & \ding{55} 
         \\
        \hline
        \begin{tabular}{@{}c@{}} $\alpha = 1+\varepsilon$  \\ undirected \end{tabular} & 
        \begin{tabular}{@{}c@{}} $\tO_{\varepsilon}(n+p\cdot n^{o(1)})$  \\ \citep{KP22,DBLP:conf/focs/ElkinS23} \end{tabular}
          & \begin{tabular}{@{}c@{}} $\Omega(n + p)$  \\  \end{tabular} &  \ding{51}  \\
        \hline 
        \begin{tabular}{@{}c@{}} $\alpha = 1+\varepsilon$  \\ directed \end{tabular} & 
        \begin{tabular}{@{}c@{}}  $\tO_{\varepsilon}(n^{2/3}p^{2/3} + np^{1/3})$ \\ \citep{KP22, BW23} \end{tabular} & \begin{tabular}{@{}c@{}}$\Omega(n^{2/(d+1)}p^{(d-1)/d})$ \\ for integer $d \ge 1$\quad \citep{CE06}\end{tabular}  & \ding{51} 
        \\
        \hline
          \begin{tabular}{@{}c@{}} $\alpha = \infty$  \\ directed \end{tabular}
          &
        \begin{tabular}{@{}c@{}} $O(n^{3/4}p^{1/2} + n^{0.59}p^{0.71} + n)$  \\ \citep{ BHT22, bodwin2024improved} \end{tabular} & 
        \begin{tabular}{@{}c@{}}$\Omega(n^{2/(d+1)}p^{(d-1)/d})$ \\ for integer $d \ge 1$\quad \citep{CE06}\end{tabular} & \ding{51} 
        \\
        \hline
    \end{tabular}
    \caption{Some previous bounds for various settings of approximate preservers on $n$-node graphs with $p$ demand pairs. We say  existing upper bounds for approximate preservers beat consistency if they imply a separation with the bounds for consistent tiebreaking schemes for some $p \in [1, n^2]$. In the second row, $\texttt{RS}(n)$ denotes the Rusza-Szemer\'{e}di function from extremal graph theory.  }
    \label{tab:distance-preserver-previous}
\end{table}

Exact distance preservers were first formally introduced in the seminal work of Coppersmith and Elkin \cite{CE06}. In their work, Coppersmith and Elkin proved upper bounds and lower bounds for exact distance preservers in several different settings. As can be seen in Table \ref{tab:distance-preserver-previous}, many of the bounds proved in \cite{CE06} are still state-of-the-art.  

A fundamental tool used in the upper bounds of \cite{CE06} is \textit{consistent tiebreaking schemes}. Given an $n$-node graph $G$, a tiebreaking scheme roughly corresponds to a collection of shortest paths $\Pi$ in $G$. Informally, we say that tiebreaking scheme $\Pi$ is \textit{consistent} if no two shortest paths $\pi_1, \pi_2 \in \Pi$ intersect, split apart, and then intersect later again (see Definition \ref{def:consistency} for a formal definition). Upper bounds for the number of distinct edges in consistent tiebreaking schemes imply upper bounds for exact distance preservers, since we may assume without loss of generality that shortest paths in the input graph are consistent. A  line of work initiated by \cite{CE06} has established tight bounds on the size of consistent tiebreaking schemes. 

\begin{theorem}[Consistent Tiebreaking Schemes] Consistent tiebreaking schemes on $n$ nodes and $p$ paths have 
\begin{itemize}
    \item $\Theta(\min(n^{1/2}p + n, np^{1/2}))$ edges in undirected graphs \cite{CE06, BV21}, and
    \item $\Theta(\min(n^{2/3}p + n, np^{1/2}))$  edges in directed graphs \cite{CE06, bodwin2017linear, BHT22}.
\end{itemize}
\end{theorem}

As we mentioned earlier, upper bounds on the number of edges in consistent tiebreaking schemes yield upper bounds for exact distance preservers, which in turn imply upper bounds for $\alpha$-approximate preservers for all  $\alpha> 1$. Recently, there has been a tremendous amount of progress in constructing approximate preservers with sparsity \textit{better} than consistent tiebreaking schemes. We summarize this progress below. 
\begin{itemize}
    \item Exact distance preservers of size $O(n^{2/3}p^{2/3} + np^{1/3})$ edges in undirected, unweighted graphs were constructed in \cite{BV21}.
    \item Almost-optimal $(1+\varepsilon)$-preservers in undirected graphs were constructed in \cite{KP22, DBLP:conf/focs/ElkinS23}. Approximate preservers of size $\tO_{\varepsilon}(n^{2/3}p^{2/3}+np^{1/3})$ edges for directed graphs were constructed in~\cite{KP22, BW23}. 
    \item Reachability preservers of size $O(n^{3/4}p^{1/2}+n^{2 - \sqrt{2}+o(1)}p^{1/\sqrt{2}} + n)$ 
 edges were constructed in \cite{BHT22, bodwin2024improved}, and \textit{online} reachability preservers of size $O(n^{0.72}p^{0.56}+n^{0.6}p^{0.7} + n)$ were constructed in \cite{bodwin2024improved}. 
\end{itemize}

All of these preserver constructions are significantly sparser than consistent tiebreaking schemes in some regime of $p$. However, as we can see in Table \ref{tab:distance-preserver-previous}, the state-of-the-art bounds for exact distance preservers in weighted (un)directed graphs do not achieve sparsity better than consistent tiebreaking. This motivates the following question. 

\begin{open}
    Is there a separation between exact distance preservers and consistent tiebreaking schemes?
    \label{open:consistency}
\end{open}

In addition to asking when we can separate approximate  preservers from consistent tiebreaking schemes, we can also ask whether we can separate different settings of $\alpha$-approximate preservers from each other. For example, the $\Omega(n^{2/3}p^{2/3})$ lower bound for exact distance preservers in weighted graphs due to \cite{CE06}, along with the recent progress on reachability preserver upper bounds due to~\cite{BHT22}, imply a separation between $\alpha$-approximate preservers when $\alpha = 1$ and $\alpha$-approximate preservers when $\alpha = \infty$. However, the following questions remained: 

\begin{open}
    Is there a separation between approximate preservers with finite stretch and approximate preservers with infinite stretch (reachability preservers)? Is there a separation between exact preservers and approximate preservers with finite stretch?
    \label{open:approx}
\end{open}

Towards answering \cref{open:consistency} and \cref{open:approx}, we prove the following results for distance preservers:

\begin{theorem}[Main Result 1]
\label{thm:main2}
    For $n$-node graphs and $p$ demand pairs, the size of the distance preserver is:
    \begin{enumerate}
             \item \label{item:thm:main2:item2}
        (\cref{thm:unw_dp}) $\widetilde{O}(n^{5/6}p^{2/3}+n)$ edges, for exact distance preservers in unweighted graphs, and        \item \label{item:thm:main2:item1} (\cref{thm:aprx_pres_lb}) $\Omega(n^{2/3}p^{2/3})$ edges, for any given finite stretch.
    \end{enumerate}
\end{theorem}

\paragraph{Remark.} At first glance, the lower bound in \cref{thm:main2} \cref{item:thm:main2:item1} appears to nearly tightly match the $\tO_{\varepsilon}(n^{2/3}p^{2/3}+np^{1/3})$ upper bound for $(1+\varepsilon)$-approximate preservers due to \cite{KP22, BW23}. However, the upper bounds of \cite{KP22, BW23} hold only for graphs with $2^{\tO(1)}$ aspect ratio, while our lower bound graph has $\text{exp}(n)$ aspect ratio. 
\\

\noindent 
\cref{thm:main2} has the following consequences:
\begin{itemize}
    \item    \cref{thm:main2} \cref{item:thm:main2:item2}  answers \cref{open:consistency} in the affirmative in the natural graph class of unweighted directed graphs. Specifically, the upper bound of \cref{thm:main2} \cref{item:thm:main2:item2} beats consistent tiebreaking in the regime of $n^{1/2} \leq p \leq n$. 
    \item \cref{thm:main2} \cref{item:thm:main2:item1} partially answers \cref{open:approx} in the affirmative for general graphs, establishing a separation between approximate preservers with finite stretch and reachability preservers. 
    This can be verified by comparing the bound in  \cref{thm:main2} \cref{item:thm:main2:item1} to the upper bound for $\alpha = \infty$ in Table \ref{tab:distance-preserver-previous}.
  
\end{itemize}

Additionally, we establish a new directed-to-undirected reduction for exact distance preservers. Let $\UDP(n,p)$ denote the best upper bound for undirected distance preservers on $n$ vertices and $p$ demand pairs, and let $\DDP(n,p)$ denotes the same for directed distance preservers.

\begin{theorem}[Main Result 2]
\label{thm:reduction}
If $\textnormal{\UDP}(n, p) = O(n^{\lambda}p^{\mu} + n)$ for constants $\lambda, \mu > 0$, then $$\textnormal{\DDP}(n, p) = O\left( n^{\frac{1}{2-\lambda}}p^{\frac{1+\mu-\lambda}{2-\lambda}} + n^{1/2}p + n\right).$$  
\end{theorem}
This reduction has the property that if we plug  the state-of-the-art $O(n^{1/2}p + n)$ upper bound   for undirected distance preservers into the reduction, we recover the state-of-the-art $O(n^{2/3}p+n)$ upper bound for directed distance preservers. 
More generally, we obtain the following consequence of the reduction. 
\begin{corollary}
Consistent tiebreaking is optimal for directed distance preservers when $p \leq n^{2/3}$ only if consistent tiebreaking is optimal for undirected distance preservers when $p \leq n$. 
\label{corr:reduction}
\end{corollary}
Our reduction makes significant progress on \cref{open:consistency}. 
If consistent tiebreaking is optimal for directed distance preservers for $p \leq n^{2/3}$, then consistent tiebreaking is optimal for undirected distance preservers for $p \leq n$. Conversely, if consistent tiebreaking is suboptimal in undirected graphs for some $p \leq n$, then consistent tiebreaking is suboptimal in directed graphs for some $p \leq n^{2/3}$. 

Along the way to proving Theorem \ref{thm:reduction}, we also obtain the following interesting algorithmic result (\cref{corr:apsp-main} in the main body), which does not seem to appear in the literature to the best of our knowledge, though it could be folklore. 
\begin{corollary}
    All-pairs shortest paths (APSP) on $n$-node, $m$-edge directed acyclic graphs can be reduced to APSP on $n$-node, $m$-edge undirected graphs in $O(m)$ time. 
    \label{corr:apsp}
\end{corollary}
\noindent A related result is the asymptotic equivalence between all-pairs shortest paths in directed dense graphs and all-pairs shortest paths in undirected dense graphs as shown by Vassilevska Williams and Williams \cite{DBLP:journals/jacm/WilliamsW18}; however, such an equivalence is not known for sparse graphs.

\subsection{ Hopsets}
\label{sec:intro:hopset}
A fundamental problem in algorithm design is how to compute (approximate) shortest paths and reachability in directed graphs efficiently. In many algorithmic settings where we want to compute shortest paths, it is useful to assume that shortest paths in the input graph contain few edges. However, in general this assumption does not hold. Instead, we can try adding a small set of additional edges to the input graph and hope that these edges reduce the lengths of shortest paths. This is the idea behind \textit{hopsets}, which we now formally define. 

\begin{definition}[Approximate Hopsets]
    Given a graph $G$, an $\alpha$-approximate hopset with hopbound $\beta$ is a set of additional edges $H$ such that:
    \begin{itemize}
        \item Every edge $(u, v) \in H$ has weight $w(u, v) = \dist_G(u, v)$.
        \item For every pair of nodes $(s, t)$ in the transitive closure of $G$, there exists an $s \leadsto t$ path $\pi$ in $G \cup H$ with at most $|\pi| \leq \beta$ edges and weight at most $w(\pi) \leq \alpha \cdot \dist_G(s, t)$.
    \end{itemize}
\end{definition}

\noindent There are three settings of approximate hopsets that have received special attention in prior work:
\begin{itemize}
    \item $\alpha = 1$ \qquad \hspace{3mm} (also known as \textit{Exact Hopsets}.)
    \item $\alpha = 1 +  \varepsilon$ \quad \hspace{0.5mm} (also known as \textit{$(1 +  \varepsilon)$-Approximate Hopsets.})
    \item $\alpha = \infty$ \qquad \hspace{1mm} (also known as \textit{Shortcut Sets}.)
\end{itemize}

\begin{table}[ht]
\small
    \centering
    \begin{tabular}{c c c c c}
    \hline 
        Setting & Size & Upper Bound & Lower Bound & Beats Folklore Sampling?  \\
        \hline 
        \hline 
        $\alpha = 1$ & $O(n)$ & 
                \begin{tabular}{@{}c@{}} $O(\sqrt{n})$ \\\citep{UY91, BRR10} \end{tabular}
          & \begin{tabular}{@{}c@{}} $\widetilde{\Omega}(\sqrt{n})$ \\\citep{BH22} \end{tabular}
          & \ding{55}
          \\
        \hline
         $\alpha = 1+\varepsilon$  
          & $O_{\varepsilon}(n)$ & 
          \begin{tabular}{@{}c@{}} $\tO(n^{1/3})$  \\ \citep{BW23} \end{tabular}
          &
           \begin{tabular}{@{}c@{}} $\Omega(n^{1/4})$ \\ \citep{BH22,DBLP:conf/soda/WilliamsXX24} \end{tabular} & \ding{51}
           \\
        \hline
        $\alpha = \infty$ & $O(n)$ & 
         \begin{tabular}{@{}c@{}} $\tO(n^{1/3})$\\ \citep{KP22a} \end{tabular}
          &
          \begin{tabular}{@{}c@{}} $\Omega(n^{1/4})$ \\ \citep{BH22,DBLP:conf/soda/WilliamsXX24} \end{tabular}& \ding{51}
        \\
        \hline
        $\alpha = \infty$ & $O(m)$ & 
         \begin{tabular}{@{}c@{}} $\tO((n^2/m)^{1/3})$ \\ \citep{KP22a} \end{tabular}
          & 
          \begin{tabular}{@{}c@{}}  $\Omega(n^{1/5})$  \\ \citep{DBLP:conf/soda/WilliamsXX24} \end{tabular} & \ding{51}
        \\
    \hline
    \end{tabular}
    \caption{Some previous bounds on various settings of $\alpha$-approximate hopsets in directed graphs. Note that the upper bound  for $\alpha = 1 + \varepsilon$ holds specifically for graphs with  aspect ratio $2^{\widetilde{O}(1)}$.} 
    \label{tab:hopset-previous}
\end{table}

Hopsets were first studied implicitly in the 90's by Ullman and Yannakakis \cite{UY91}, and defined explicitly by Cohen \cite{Cohen00}. In particular, Ullman and Yannakakis gave an implicit construction for exact hopsets using a simple sampling scheme, which has since become folklore (and slightly improved \cite{BRR10}). 

\begin{theorem}[Folklore Sampling \cite{UY91, BRR10}]
Every $n$-node graph has an $O(n)$-size hopset  with hopbound $O(n^{1/2})$.
\label{thm:folklore}
\end{theorem}

While many algorithmic applications for hopsets have been found since the 90's, no progress had been made on approximate hopset upper bounds in directed graphs until recently. In 2003, Hesse  gave the first polynomial lower bounds for $O(n)$-size and $O(m)$-size shortcut sets, proving that indeed polynomial hopbound is necessary for hopsets \cite{Hesse03}. Nevertheless, it remained open whether the folklore sampling algorithm could be improved, even in the easiest setting of $\alpha = \infty$.

In a  breakthrough result by Kogan and Parter \cite{KP22a}, it was shown that there exist $O(n)$-size shortcut sets ($\infty$-approximate hopsets) with hopbound $\tO(n^{1/3})$. 
Kogan and Parter also proved that there exist $O(n)$-size $(1+\varepsilon)$-approximate hopsets with hopbound $\widetilde{O}_{\varepsilon}(n^{2/5})$. The hopbound for $(1+\varepsilon)$-approximate hopsets was  subsequently improved to $\widetilde{O}_{\varepsilon}(n^{1/3})$ by Bernstein and Wein \cite{BW23}, using a careful analysis of approximate shortest paths in directed graphs.\footnote{$\tO$ hides $\polylog(n)$ factors.} 

The possibility of extending these improvements to exact hopsets was ruled out by \cite{BH22}, where it was proved that $O(n)$-size hopsets have hopbound $\widetilde{\Omega}(n^{1/2})$ in the worst case. This lower bound proved that folklore sampling is indeed optimal for exact hopsets. Additionally, this lower bound, along with the upper bounds of \cite{KP22a, BW23}, established a separation between exact hopsets and approximate hopsets. 


\noindent However, the following question remained:

\begin{open}
 \label{open:open1}
    Is there a separation between approximate hopsets with finite stretch and approximate hopsets with infinite stretch (shortcut sets)?
\end{open}

One of the barriers to making progress on this question is that all known lower bounds for hopsets follow from existing lower bounds for shortcut sets (with the exception of \cite{BH22, KP22}). Indeed, while there has been a long series of work improving lower bounds for shortcut sets~\cite{Hesse03,HuangP21, LVWX22, BH22, DBLP:conf/soda/WilliamsXX24}, there has been relatively little progress in proving lower bounds for hopsets that surpass existing lower bounds for shortcut sets. This barrier motivates the following question.
\begin{open}
 \label{open:open2}
Can we prove new lower bounds for hopsets that surpass existing lower bounds for shortcut sets? 
\end{open}

Towards answering these open questions, we show a range of new lower bounds for various settings of hopsets, summarized in \cref{thm:main1} below:

\begin{theorem}[Main Result 3]
\label{thm:main1}
For hopsets in directed graphs, the hopbound is:
\begin{enumerate}
    \item (\cref{thm:dir_apx_hop_lb}.) \label{item:thm:main1:item1} $\Omega(n^{1/2})$ for $O(n)$-size approximate hopsets with any given finite stretch,
    \item (\cref{thm:shortcut}.) \label{item:thm:main1:item2} $\Omega(n^{2/9})$ for $O(m)$-size hopsets and stretch infinity (shortcut sets), and
     \item (\cref{thm:hopset}.) \label{item:thm:main1:item3} $\Omega(n^{2/7})$, for $O(m)$-size exact hopsets, in unweighted graphs (this bound holds even for undirected graphs). 
\end{enumerate}
\end{theorem}

\paragraph{Remark.} At first glance, the $\Omega(n^{1/2})$ lower bound in \cref{item:thm:main1:item1} of \cref{thm:main1} for $O(n)$-size approximate hopsets seems to contradict the $O(n^{1/3})$ upper bound for $\tO_{\varepsilon}(n)$-size $(1+\varepsilon)$-approximate hopsets due to \cite{BW23}. However, the upper bounds of \cite{BW23} hold only for graphs with $2^{\widetilde{O}(1)}$ aspect ratio. Our lower bound graph  has $\text{exp}(n)$ aspect ratio. 
\\

\noindent \cref{thm:main1} has the following consequences:
\begin{itemize}
    \item \cref{thm:main1} \cref{item:thm:main1:item1} answers \cref{open:open1} in the affirmative for general graphs. In particular, it shows that for graphs with exponential aspect ratio, the folklore sampling $O(\sqrt{n})$ bound of \cite{UY91, BRR10} is optimal.

    \item \cref{thm:main1} \cref{item:thm:main1:item1} also separates approximate hopsets in graphs with polynomial aspect ratio from approximate hopsets in graphs with exponential aspect ratio. This is highly unusual behavior in network design problems. For example, in undirected graphs the state-of-the-art bounds for multiplicative spanners \cite{ADDJS93}, approximate distance preservers \cite{KP22}, and hopsets \cite{EN19, HP19, KP22} have no dependency on the aspect ratio of the input graph.
    
    \item \cref{thm:main1} \cref{item:thm:main1:item2} improves the previous $\Omega(n^{1/5})$ lower bound  for $O(m)$-size shortcut sets  from \cite{DBLP:conf/soda/WilliamsXX24} to $\Omega(n^{2/9})$.

    \item \cref{thm:main1} \cref{item:thm:main1:item3} makes progress on \cref{open:open2}. The best known lower bound for $O(m)$-size shortcut is $\Omega(n^{2/9})$, as shown in \cref{item:thm:main1:item2}. Our bound in \cref{item:thm:main1:item3} beats the current best bound for $O(m)$-size shortcut, establishing another example where the current best lower bound for some setting of hopset is higher than that for shortcut set. In fact, our $\Omega(n^{2/7})$ lower bound is even higher than the $\Omega(n^{1/4})$ lower bound for $O(n)$-size shortcut set \cite{BH22, DBLP:conf/soda/WilliamsXX24}.
    Improving existing lower bounds for $O(n)$-size shortcut set beyond $\Omega(n^{1/4})$ 
    is likely to be difficult due to known barriers for improving reachability preserver lower bounds \cite{BHT22}. 
\end{itemize}

\section{Technical Overview}

In this section, we provide intuition and proof overviews for our results. 
\begin{itemize}
    \item In Section \ref{sec:reduction_tech}, we discuss several new insights about the (non)existence of cycles in systems of shortest paths in directed graphs, with applications to exact distance preservers and APSP. (This corresponds to Item \ref{item:thm:main2:item2} of Theorem \ref{thm:main2} and  Theorem \ref{thm:reduction}.)
    \item In Section \ref{sec:apx_hop_tech}, we discuss the new weighting scheme we developed to prove new lower bounds for approximate hopsets and preservers (Item \ref{item:thm:main1:item1} of Theorem \ref{thm:main1} and Item \ref{item:thm:main2:item1} of Theorem \ref{thm:main2}). 
    \item In Section \ref{sec:other_tech}, we discuss our new lower bound constructions for $O(m)$-size shortcut sets and $O(m)$-size exact hopsets in unweighted graphs (Items \ref{item:thm:main1:item2} and \ref{item:thm:main1:item3} of Theorem \ref{thm:main1}).
\end{itemize}

\subsection{Exact Preserver Reductions and Upper Bounds}
\label{sec:reduction_tech}

In this section, we discuss the ideas behind our new reduction from directed distance preservers to undirected distance preservers, as well as our new upper bound for  distance preservers in directed unweighted graphs. Both of these results stem from new insights about the existence of acyclic \textit{and cyclic} structures in systems of shortest paths in directed graphs. These insights are summarized in Observation \ref{obs:undirected} and Lemma \ref{obs:directed}, respectively. 

\paragraph{Directed to Undirected Reduction.}
We will briefly sketch the strategy we use to obtain an extremal reduction from directed to undirected distance preservers. Let $G$ be an $n$-node directed graph, and let $P$ be a set of $|P| = p$ demand pairs in $G$. Our reduction from directed distance preservers to undirected distance preservers happens in two steps:
\begin{enumerate}
    \item First, we give an extremal reduction from directed distance preservers to distance preservers in DAGs. 
    \item Second, we give an extremal reduction from distance preservers in DAGs to distance preservers in undirected graphs. 
\end{enumerate}

The key idea behind our first reduction from directed distance preservers to distance preservers in DAGs is summarized in Observation \ref{obs:undirected} and visualized in Figure \ref{fig:dag}.

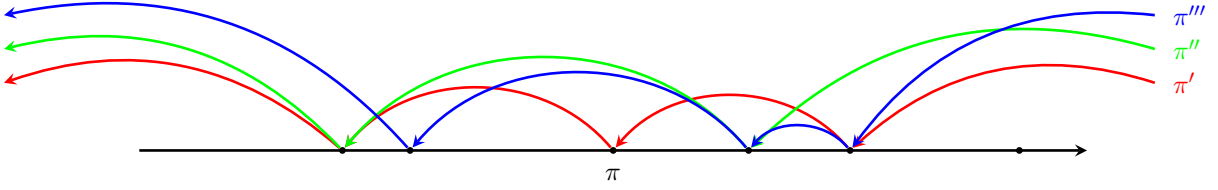
\begin{figure}[ht]
    \centering
    \begin{tikzpicture}[scale=0.9, transform shape]
		\node at (7,0) [label = below:$\pi$] (pi){};

		\node at (3, 0) [circle, fill, inner sep = 1pt] (v1){};
		\node at (7, 0) [circle, fill, inner sep = 1pt] (v2){};
		\node at (9, 0) [circle, fill, inner sep = 1pt] (v3){};
		\node at (10.5, 0) [circle, fill, inner sep = 1pt] (v4){};
		\node at (13, 0) [circle, fill, inner sep = 1pt] (v5){};
		\node at (4, 0) [circle, fill, inner sep = 1pt] (v6){};
	
    \draw [-stealth, line width = 1pt] (0, 0) to[] (14, 0);
    
    \node at (15,1) [label = {[text=red]right:$\pi'$}] (pi'){};
	\draw [-stealth, line width = 1pt, bend right, color = red] (15, 1) to[] (v4);
	\draw [-stealth, line width = 1pt, out = 130, in = 50, color = red] (v4) to[] (v2);
	\draw [-stealth, line width = 1pt, out = 130, in = 50, color = red] (v2) to[] (v1);
	\draw [-stealth, line width = 1pt, bend right, color = red] (v1) to[] (-2, 1);
	
	\node at (15,1.5) [label = {[text=green]right:$\pi''$}] (pi''){};
	\draw [-stealth, line width = 1pt, bend right, color = green] (15, 1.5) to[] (v3);
	\draw [-stealth, line width = 1pt, out = 130, in = 50, color = green] (v3) to[] (v1);
	\draw [-stealth, line width = 1pt, bend right, color = green] (v1) to[] (-2, 1.5);
	
	\node at (15,2) [label = {[text=blue]right:$\pi'''$}] (pi'''){};
	\draw [-stealth, line width = 1pt, bend right, color = blue] (15, 2) to[] (v4);
	\draw [-stealth, line width = 1pt, out = 130, in = 50, color = blue] (v4) to[] (v3);
	\draw [-stealth, line width = 1pt, out = 130, in = 50, color = blue] (v3) to[] (v6);
	\draw [-stealth, line width = 1pt, bend right, color = blue] (v6) to[] (-2, 2);

    \end{tikzpicture}
    \caption{This figure depicts an example of a path system $\Pi = \{\pi, \pi', \pi'', \pi'''\}$. The paths in $\Pi \setminus \{\pi\}$, restricted to the nodes in path $\pi$, form a DAG. }
    \label{fig:dag}
\end{figure}

\begin{observation}[DAG Structure (Informal)]
\label{obs:undirected}
Let $\Pi$ be a \textit{consistent} collection of shortest paths in a directed graph $G$, and let $\pi$ be a path in $\Pi$. Then the paths in $\Pi \setminus \{\pi\}$, restricted to the nodes in path $\pi$, form a DAG.
\end{observation}
In Observation \ref{obs:undirected}, we define the restriction  of  a path $\pi' \in \Pi \setminus \pi$ to the nodes in path $\pi$ as the path $\pi' \cap \pi$, with nodes ordered in the sequence they appear in $\pi'$. Then, as we can visualize in Figure \ref{fig:dag}, Observation \ref{obs:undirected} tells us that the paths in $\Pi$, restricted to the nodes in path $\pi$, form a DAG with a topological ordering that is the reverse of the order of nodes in $\pi$. 
We note that Observation \ref{obs:undirected} is not strictly true, unless paths in $\Pi$ are pairwise edge-disjoint. However, edge-disjointness holds without loss of generality for shortest paths in distance preservers, as we discuss at the beginning of Section \ref{sec:reduction}.

Here is an outline of how we can use Observation \ref{obs:undirected} to reduce extremal bounds for directed distance preservers to extremal bounds for distance preservers in DAGs:
\begin{itemize}
    \item Let $H$ be a minimal distance preserver of directed graph $G$ and demand pairs $P$.
    \item If $H$ is dense, then there exists a shortest $s \leadsto t$ path $\pi$ in $H$, for some $(s, t) \in P$, such that $\pi$ contains many high-degree nodes. 
    \item Let $D$ be the DAG in Observation \ref{obs:undirected} corresponding to the shortest paths in $H$ induced on path $\pi$. Then since path $\pi$ contains many high-degree nodes, DAG $D$ is dense. 
    \item We can apply extremal bounds for distance preservers in DAGs to obtain an upper bound on $|E(D)|$. This will in turn yield an upper bound on $|E(H)|$. 
\end{itemize}
This completes the sketch of our reduction from directed distance preservers to distance preservers in DAGs; the full reduction is given in Section \ref{sec:dir_to_dag}. 

The second step of our reduction is an extremal reduction from distance preservers in DAGs to distance preservers in undirected graphs. This reduction is simple, but to our knowledge previously unknown. Here is a rough sketch of how we can convert shortest paths in a DAG to shortest paths in an undirected graph: Let $v_1, \ldots, v_n$ be the topological order of the DAG, and let $W$ be a sufficiently large number (say larger than the sum of all edge weights in the DAG). Then we add $W \cdot (j - i)$ to the weight of any edge $(v_i, v_j)$ in the graph, and make the edge undirected.

A simple analysis of the undirected reweighted graph shows that shortest paths remain the same in this new graph. This reduction immediately implies that the extremal size of distance preservers in DAGs is at most the extremal size of distance preservers in undirected graphs. Additionally, this reduction implies a fast \textit{algorithmic} reduction from APSP in DAGs to APSP in undirected graphs. We give a complete presentation of our reduction and results in Section \ref{sec:dag_to_und}.

\paragraph{Exact Preserver Upper Bound.}
Our exact distance preserver upper bound (\cref{item:thm:main1:item2} in \cref{thm:main2}) for directed unweighted graphs emerges from the following intuition. Let $G$ be an $n$-node directed, unweighted graph, and let $P$ be a set of $|P| = p$ demand pairs in $G$. If graph $G$ is a DAG, then we can apply our extremal reduction from distance preservers in DAGs to distance preservers in undirected graphs. This implies that $G, P$ has a distance preserver of size $O(\min\{n^{1/2} p + n, np^{1/2}\})$. 

Then if graph $G$ has $|E(G)| = \omega(n^{1/2}p+n)$ edges,  we know that $G$ must contain a directed cycle.  More generally, if $|E(G)| \gg n^{1/2}p+n$, then $G$ contains \textit{many} short directed cycles. This fact provides the intuition behind the following (informal) lemma. 

\begin{lemma}[Dense Low-Diameter Cluster Lemma (cf. Lemma \ref{lem:low-diam-scc})]
\label{obs:directed}
    Let $H$ be a minimal distance preserver of directed, unweighted graph $G$ and demand pairs $P$. If $H$ has (sufficiently large) average degree $d$, then there exists a set $S \subseteq V(G)$ such that
    \begin{enumerate}
        \item $|S| = \Theta(d)$,
        \item $\deg_H(v) = \Omega(d)$ for all $v \in S$, and
        \item Set $S$ is strongly connected and has small (weak) diameter in $H$. 
    \end{enumerate}
\end{lemma}
Lemma \ref{obs:directed} can be understood as a directed generalization of the following standard lemma in undirected graphs:
\begin{lemma}[\cite{diestel2000graduate}, cf. \cite{BV21}]
    Let $H$ be an undirected, unweighted graph with average degree $d$. Then $H$ has a set $S$ of $|S| = \Theta(d)$ nodes, where $\deg_H(v) = \Omega(d)$ for all $v \in S$ and $S$ has strong diameter 2. 
    \label{lem:und_lem}
\end{lemma}
Note that we cannot hope for a generalization of  Lemma \ref{lem:und_lem} to general directed graphs because directed graphs can be arbitrarily dense and acyclic (e.g., the directed biclique). We formalize Lemma \ref{obs:directed} in Lemma \ref{lem:low-diam-scc} and we prove it in Section \ref{subsec:low_diam_scc}. 

Once we prove the Dense Low-Diameter Cluster Lemma, our distance preserver upper bounds in directed, unweighted graphs follow from a straightforward simulation of the upper bound argument for  distance preservers in undirected, unweighted graphs due to \cite{BV21}. The key subroutine of the  upper bound in \cite{BV21} is a sourcewise distance preserver construction, using the set $S$ in Lemma \ref{lem:und_lem} as the set of source nodes.
The authors of \cite{BV21} prove that there exists a particularly sparse sourcewise distance preserver of demand pairs $P \subseteq S \times V(G)$,  using a clever application of the pigeonhole principle to shortest path distances. In our upper bound we follow this framework, while using the Dense Low-Diameter Cluster Lemma instead of  Lemma \ref{lem:und_lem}.

\subsection{Approximate Hopset and Preserver Lower Bounds}
\label{sec:apx_hop_tech}

Our lower bounds for approximate hopsets and preservers emerge from the following observation. 

\begin{observation}[Informal]
    Existing lower bounds for exact preservers and hopsets in undirected graphs can be converted into lower bounds for $\alpha$-approximate preservers and hopsets in directed graphs for any finite $\alpha$. 
\end{observation}

We convert lower bounds for exact distance problems in undirected graphs into lower bounds for approximate distance problems in directed graphs using a simple argument. Given an undirected lower bound graph, we do the following:
\begin{enumerate}
    \item Direct the edges of the lower bound graph in a natural way.
    \item Reweight the edges of the lower bound graph using the \textit{Over-Under} weighting scheme.
\end{enumerate}

We now illustrate the Over-Under weighting scheme,  by sketching  the proof of our directed approximate preserver lower bound. Our application of the Over-Under weighting scheme to directed approximate hopset lower bounds follows a similar argument. 

\paragraph{The Over-Under Weighting Scheme for Approximate Preserver Lower Bounds.} 
For exact distance preservers in undirected weighted graphs, the best known size lower bound is $\Omega(n^{2/3}p^{2/3})$  due to a construction by \cite{CE06}. The lower bound graph of \cite{CE06} corresponds to an arrangement of points and lines in the plane with many point-line incidences. Formally, there exists an arrangement of points $\mathcal{P} \subseteq \mathbb{R}^2$ and lines $\mathcal{L}$ in the plane such that:
\begin{enumerate}
    \item $|\mathcal{P}| = n$ and $|\mathcal{L}| = p$,
    \item There are $\Omega(n^{2/3}p^{2/3})$ distinct point-line incidences between points in $\mathcal{P}$ and lines in $\mathcal{L}$,
    \item Any two distinct lines $L, L' \in \mathcal{L}$ intersect on at most one point, and 
    \item Every line $L \in \mathcal{L}$ is a unique shortest path in $\mathbb{R}^2$ under the Euclidean metric. 
\end{enumerate}
The set of points $\mathcal{P}$ and set of lines $\mathcal{L}$ can naturally be converted into an $n$-node graph $G$ embedded in the Euclidean plane and a set of demand pairs $P$ of size $|P| = p$ as follows:
\begin{itemize}
    \item Let $V(G) = \mathcal{P}$,
    \item Add edge $(x, y) \in \mathcal{P} \times \mathcal{P}$ to $E(G)$ if points $x$ and $y$ occur consecutively on some line $L \in \mathcal{L}$,
    \item Assign edge $(x, y)$ weight $w(x, y) = \|x - y\|$, where $\| \cdot \|$ denotes Euclidean distance, and
    \item Add demand pair $(s, t)\in \mathcal{P} \times \mathcal{P}$ to $P$ if $s$ and $t$ are the first and last points in $\mathcal{P}$  on some line $L \in \mathcal{L}$, respectively.
\end{itemize}
Using Properties 1-4 of the arrangement of points $\mathcal{P}$ and lines $\mathcal{L}$, it becomes a simple exercise to prove that  any exact preserver of $G, P$ has at least $\Omega(n^{2/3}p^{2/3})$ edges, recovering the \cite{CE06} lower bound. 
However, our goal is to reweight graph $G$ and  make it directed, so that any $\alpha$-approximate preserver of $G, P$ requires $\Omega(n^{2/3}p^{2/3})$ edges. Here is a rough sketch of how we do this:
\begin{enumerate}
    \item For every undirected edge $(x, y) \in E(G)$, if point $x$ has a smaller first coordinate than point $y$, then convert edge $(x, y)$ into directed edge $x \rightarrow y$. 
    \item Let $\mathcal{L} = \{L_1, \dots, L_p\}$ be the lines in $\mathcal{L}$ written by increasing slope. We assign the edges in $E(G)$ corresponding to line $L_i$ weight $(\alpha \cdot n)^i$.
\end{enumerate}

We now give some intuition for why this yields a lower bound against $\alpha$-approximate preservers. 
Consider any line $L_i \in \mathcal{L}$, and let $x_1, \dots, x_k$ be the points in $\mathcal{P}$ on $L_i$ ordered by increasing first coordinate. Then $(x_1, x_k) \in P$ is a demand pair in $P$. We claim that path $\pi_{L_i} = (x_1, \dots, x_k)$ is the unique $\alpha$-approximate $x_1 \leadsto x_k$ path in $G$. This claim is a consequence of the following property of directed paths in $G$, which can be visualized in Figure \ref{fig:over_under}.
\begin{observation}[Over-Under Property]
    Any $x_1 \leadsto x_k$ path $\pi$ such that $\pi \not = \pi_{L_i}$ must contain an edge lying on a line with slope larger than $L_i$.
\end{observation}
Using the Over-Under Property, it follows that any $x_1 \leadsto x_k$ path $\pi \neq \pi_{L_i}$ must contain an edge $e$ lying on a line with slope larger than $L_i$. However, by our exponential weighting scheme, $$w(e) \geq (\alpha \cdot n)^{i+1} >  \alpha \cdot w(\pi_{L_i}).$$ 
Therefore, any $\alpha$-approximate preserver of $G, P$ must contain every edge in $\pi_{L_i}$. This directly implies an $\Omega(n^{2/3}p^{2/3})$ lower bound for any $\alpha$-approximate preserver of $G, P$, completing the sketch of \cref{item:thm:main2:item1} of \cref{thm:main2}. 

The Over-Under weighting scheme is also used in our construction for approximate hopset lower bounds, combined with the previous undirected exact hopset lower bound of \cite{bodwin2023folklore} and simplifications from \cite{DBLP:conf/soda/WilliamsXX24}.

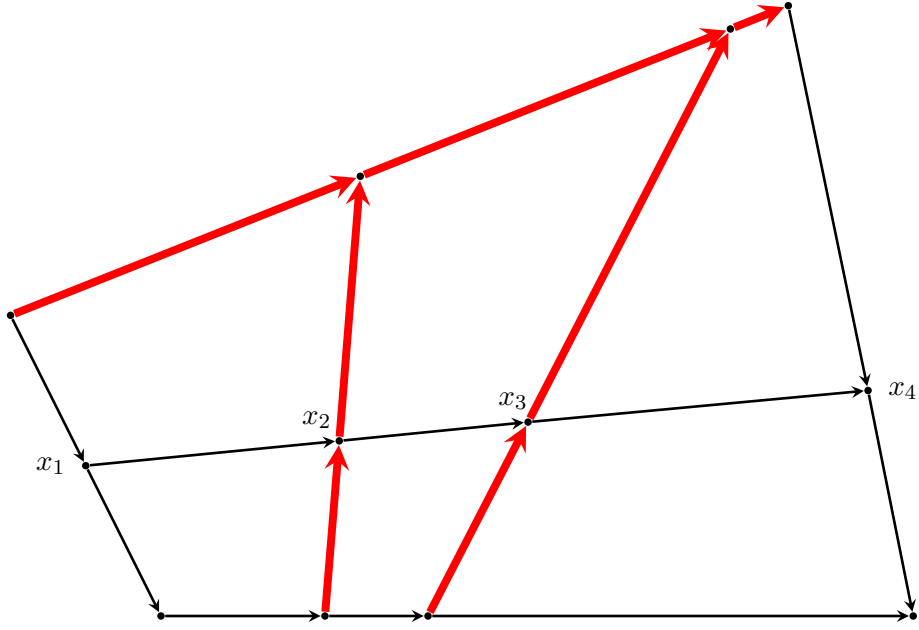
\begin{figure}[ht]
    \centering
    \begin{tikzpicture}
		\node at (0,0) [circle, fill, inner sep = 1pt] (A){};
		\node at (10,0) [circle, fill, inner sep = 1pt] (B){};
		\node at (-1,2) [circle, fill, inner sep = 1pt] (C){};
		\node at (9.4,3) [circle, fill, inner sep = 1pt] (D){};
		\node at (-2,4) [circle, fill, inner sep = 1pt] (E){};
		\node at (8.34, 8.12) [circle, fill, inner sep = 1pt] (F){};
		\node at (2.65, 5.85) [circle, fill, inner sep = 1pt] (G){};
		\node at (2.18, 0) [circle, fill, inner sep = 1pt] (H){};
		\node at (7.57, 7.81) [circle, fill, inner sep = 1pt] (I){};
		\node at (4.88, 2.58) [circle, fill, inner sep = 1pt] (J){};
		\node at (3.55, 0) [circle, fill, inner sep = 1pt] (K){};
		\node at (2.37, 2.33) [circle, fill, inner sep = 1pt] (L){};
		
		\node at (C) [label = left:$x_1$] (){};
  		\node at (L) [label = {[xshift=-0.3cm, yshift=-0.1cm]$x_2$}] (){};
    	\node at (J) [label = {[xshift=-0.2cm, yshift=-0.1cm]$x_3$}] (){};
		\node at (D) [label = right:$x_4$] (){};

		\draw [-stealth, line width = 1pt] (F) to[] (D);
		\draw [-stealth, line width = 1pt] (D) to[] (B);
		
		\draw [-stealth, line width = 1pt] (E) to[] (C);
		\draw [-stealth, line width = 1pt] (C) to[] (A);
				
		\draw [-stealth, line width = 1pt] (A) to[] (H);
		\draw [-stealth, line width = 1pt] (H) to[] (K);
		\draw [-stealth, line width = 1pt] (K) to[] (B);
		
		\draw [-stealth, line width = 1pt] (C) to[] (L);
		\draw [-stealth, line width = 1pt] (L) to[] (J);
		\draw [-stealth, line width = 1pt] (J) to[] (D);
		
		\draw [-stealth, line width = 3pt, color = red] (E) to[] (G);
		\draw [-stealth, line width = 3pt, color = red] (G) to[] (I);
		\draw [-stealth, line width = 3pt, color = red] (I) to[] (F);
		
		\draw [-stealth, line width = 3pt, color = red] (K) to[] (J);
		\draw [-stealth, line width = 3pt, color = red] (J) to[] (I);
		
		\draw [-stealth, line width = 3pt, color = red] (H) to[] (L);
		\draw [-stealth, line width = 3pt, color = red] (L) to[] (G);
		
    \end{tikzpicture}
    \caption{This figure depicts an example graph constructed by a point-line incidence system. All edges on lines with a larger slope than the line containing $x_1$ and $x_4$ are marked as red, wider lines. Except the path $(x_1, x_2, x_3, x_4)$, all other paths from $x_1$ to $x_4$ use red edges, which have weight much larger than the edges on $(x_1, x_2, x_3, x_4)$. }
    \label{fig:over_under}
\end{figure}

\subsection{\texorpdfstring{$O(m)$}{O(m)}-size Shortcut Set and Hopset Lower Bounds}

Here we give an overview view of our $O(m)$-size shortcut set and hopset lower bounds. 

\label{sec:other_tech}

\paragraph{Shortcut Set Lower Bound.} Our construction for an $\Omega(n^{2/9})$ diameter lower bound for  $O(m)$-size shortcut sets is a refinement of a prior shortcut set lower bound due to \cite{DBLP:conf/soda/WilliamsXX24}. The construction of~\cite{DBLP:conf/soda/WilliamsXX24} obtains an $\Omega(n^{1/5})$ lower bound for $O(m)$-size shortcut sets. Their constructed graph is a layered graph with $\Theta(n^{1/5})$ layers, and each layer has $\Theta(n^{4/5})$ vertices. Our construction intuitively can be viewed as a concatenation of $n^{1/5}$ copies of their graphs, where the last layer of the previous copy is identified with the first layer of the next copy. Due to technical reasons, the number of vertices in each layer needs to blow up by a factor of $\Theta((n^{1/5})^3)$. Thus, the total number of vertices is $\Theta(n^{9/5})$, and the number of layers is $\Theta(n^{2/5})$, which intuitively gives an $\Omega(n^{2/9})$ lower bound. 

\paragraph{Unweighted Exact Hopset Lower Bound.} Our construction for unweighted exact hopset lower bounds is an adaptation of our layered construction for the $\Omega(n^{2/9})$ lower bound for $O(m)$-size shortcut sets. The high-level intuition for the higher $\Omega(n^{2/7})$ lower bound is the following: in the construction for shortcut sets, we need the critical paths to be the unique paths between their start and end nodes; in comparison, in hopsets, we only need the critical paths to be the unique \textit{shortest} paths between their start and end nodes. This allows us to create an unlayered version of shortcut set construction, which reduces the total number of vertices in the graph, and leads to a higher lower bound.

\section{Preliminaries}
In this section, we present the notations and definitions used in our paper.
\subsection*{Notation}
\begin{itemize}
    \item For integers $n \ge 0$, we use $[n]$ to denote $\{1, \ldots, n\}$ and $\ints{n}$ to denote $\{0, \ldots, n - 1\}$. For integers $i, j \geq 0$, where $i < j$, we use $[i, j]$ to denote $\{i, i+1, \dots, j-1, j\}$. 
    \item Let  $G = (V, E, w)$ be a weighted graph with weight function $w: E \rightarrow \mathbb{R}$. For vertex $v \in V$, we use $\deg_G(v)$ to denote the number of edges incident to $v$ in $G$. For $s, t \in V$, we use $w(s, t)$ to denote the edge weight of the edge from $s$ to $t$, $\dist_G(s, t)$ to denote the distance from $s$ to $t$ in $G$ and $\pi_G(s, t)$ to denote a shortest path from $s$ to $t$. The subscripts $G$ might be dropped if it is clear from context. 
    \item For $n, p \ge 1$, we use $\textsc{UDP}(n, p)$ (resp. $\textsc{DDP}(n, p)$, $\textsc{DADP}(n, p)$) to denote the best bound on the size of exact distance preservers for $n$-node undirected (resp. directed, directed acyclic) weighted graphs with $p$ demand pairs. 
    \item For a vertex $v \in V(G)$, we use $\text{outdeg}_G(v)$ to denote the number of outgoing edges incident to $v$ in $G$, and we use $\text{indeg}_G(v)$ to denote the number of incoming edges incident to $v$ in $G$. We let $\deg_G(v) = \text{indeg}_G(v) + \text{outdeg}_G(v)$.  

    \item We will frequently use symbol $\pi$ to denote paths. For a simple path $\pi$, $|\pi|$ denotes its number of edges and $w(\pi)$ denotes the total edge weights of edges on $\pi$. For two vertices $x, y$ on $\pi$ where $x$ appears before $y$, $\pi[x, y]$ or $\pi[x \leadsto y]$ denotes the subpath of $\pi$ from $x$ to $y$. 
    
    \item Given a path $\pi$ and a set of nodes $S$, we use $\pi \cap S$ to denote the subpath of $\pi$ induced on the nodes of $S$. Formally, $\pi \cap S$ is the path obtained by listing the nodes in $V(\pi) \cap S$ in the order they appear on path $\pi$. 

    \item Given a graph $G$, a set of demand pairs $P \subseteq V(G) \times V(G)$, and a distance preserver $H$ of $G, P$, we say that $H$ is a minimal distance preserver if no edge can be deleted from $H$ while ensuring that $H$ remains a distance preserver of $G, P$. If shortest paths in $G$ are unique shortest paths, then there is a unique minimal distance preserver $H$ of $G$. 
\end{itemize}

\subsection*{Definitions}

\begin{definition}[Tiebreaking Scheme]
    Given a graph $G$ and a set of demand pairs $P \subseteq V(G) \times V(G)$, we define a tiebreaking function $\pi(\cdot, \cdot)$ with respect to $G, P$ as follows. For each demand pair $(s, t) \in P$, fix a shortest path $\pi_{s, t}$ in $G$. Then we define the output of tiebreaking function $\pi(\cdot, \cdot)$ on input $(s, t)$ to be $$\pi(s, t) = \pi_{s, t}$$
    for all $(s, t) \in P$. Additionally, we let $\pi(P)$ denote the image of $P$ under $\pi(\cdot, \cdot)$, i.e., $\pi(P)$ denotes the collection of paths $\{\pi(s, t) \mid (s, t) \in P\}$. Likewise, we let
    $$
    \cup_{\pi \in \pi(P)} \pi
    $$
    denote the graph obtained by taking the (non-disjoint) union of the paths in set $\pi(P)$. 
\end{definition}

\begin{definition}[Consistent Tiebreaking Scheme]
\label{def:consistency}
Let $\pi(\cdot, \cdot)$ be a tiebreaking scheme with respect to $G, P$. Then $\pi(\cdot, \cdot)$ is consistent if 
for all paths $\pi_1, \pi_2 \in \pi(P)$ and nodes $x, y \in V(G)$ such that $x$ precedes $y$ in both $\pi_1$ and $\pi_2$, then we have that $\pi_1[x \leadsto y] = \pi_2[x \leadsto y]$. 
\end{definition}

\begin{definition}[Branching Events \cite{CE06}]
Let $G$ be a directed graph, let $P$ be a set of demand pairs, and let $\pi(\cdot, \cdot)$ be a tiebreaking scheme. Let $H = \cup_{(s, t) \in P} \pi(s, t)$ be the exact distance preserver implied by $\pi(\cdot, \cdot)$. 

A branching event is a pair of directed edges $(x, y), (x, z) \in E(H)$ in $H$, for any $x, y, z \in V(G)$. For each branching event $(x, y), (x, z)$, we associate a pair of paths $\pi_1, \pi_2 \in \pi(P)$ such that $(x, y) \in \pi_1$ and $(x, z) \in \pi_2$. If there is more than one choice of paths $\pi_1, \pi_2$, then we choose arbitrarily.
\label{def:be}
\end{definition}
We say that a branching event $(x, y), (x, z)$ is \textit{between} paths $\pi_1, \pi_2 \in \pi(P)$ if paths $\pi_1, \pi_2$ are associated with branching event $(x, y), (x, z)$ as in Definition \ref{def:be}.

\section{Directed to Undirected Preserver Reduction}
\label{sec:reduction}

The goal of this section will be to prove Theorem \ref{thm:reduction} and Corollaries \ref{corr:reduction} and \ref{corr:apsp}. We restate Theorem \ref{thm:reduction} below.

\begin{theorem*}
If $\textnormal{\UDP}(n, p) = O(n^{\lambda}p^{\mu} + n)$ for constants $\lambda, \mu > 0$, then $$\textnormal{\DDP}(n, p) = O\left( n^{\frac{1}{2-\lambda}}p^{\frac{1+\mu-\lambda}{2-\lambda}} + n^{1/2}p + n\right).$$  
\end{theorem*}

We will need several intermediate claims and lemmas before we can prove Theorem \ref{thm:reduction}. Let $G, P$ be an $n$-node directed, weighted graph and  a set of $|P| = p$ demand pairs, such that the minimal distance preserver $H$ of $G, P$ has $$|E(H)| = \Theta(\textsc{DDP}(n, p))$$ edges. Let $\pi(\cdot, \cdot)$ be a tiebreaking scheme associated with $G, P$. We will make the following assumptions on $G$, $P$, $H$,   and $\pi(\cdot, \cdot)$, which we can assume without loss of generality:
\begin{enumerate}
    \item Path $\pi(s, t)$ is a unique shortest $ s\leadsto t$ path in $G$ for all $s, t \in V(G)$. In particular, this implies that $\pi(\cdot, \cdot)$ is consistent.
    \item Paths in $\pi(P)$ are pairwise edge-disjoint.
    \item Every node $v \in V(G)$ has degree at least  $\deg_H(v) \geq \frac{|E(H)|}{4n}$ in $H$.
    \item Every path $\pi \in \pi(P)$ has at least $|\pi| \geq \frac{|E(H)|}{4p}$ edges. 
\end{enumerate}
Assumptions 1 and 2 hold without loss of generality due to the Independence Lemma for weighted distance preservers (Lemma 33 of \cite{BHT22}). 
Note that Assumption 1 implies that tiebreaking scheme $\pi(\cdot, \cdot)$ is consistent. Assumptions 3 and 4 follow from standard arguments in this area, so we defer the proof of these assumptions to Appendix \ref{app:ass}.

Our extremal reduction from $\textsc{DDP}(n, p)$ to $\textsc{UDP}(n, p)$ will happen in two steps. First, in Section \ref{sec:dir_to_dag}, we will reduce $\textsc{DDP}$ to $\textsc{DADP}$, culminating in the proof of Lemma \ref{lem:dir_dag}. Second, in Section \ref{sec:dag_to_und} we will show that $\textsc{DADP}(n, p) \leq \textsc{UDP}(n, p)$ and prove Corollary \ref{corr:apsp}.  We combine these two reductions in Section \ref{subsec:finishing_reduction} to complete the proofs of Theorem \ref{thm:reduction} and Corollary \ref{corr:reduction}. 

\subsection{Directed to DAG Reduction}

We begin our reduction from $\textsc{DDP}$ to $\textsc{DADP}$ with the following observation about the structure of paths in $\pi(P)$. 

\label{sec:dir_to_dag}

\begin{claim}[Path Intersection Claim]
    Fix two distinct shortest paths $\pi_1, \pi_2 \in \pi(P)$. Let $$v_1, \dots, v_k \in V(\pi_1) \cap V(\pi_2)$$ be the nodes in $V(\pi_1) \cap V(\pi_2)$, listed in the order they appear in path $\pi_1$. Then nodes in $V(\pi_1) \cap V(\pi_2)$ appear on path $\pi_2$ in the reverse order $v_k, \dots, v_1$. 
    \label{clm:path_inter}
\end{claim}
\begin{proof}
    Suppose for the sake of contradiction that there are two vertices $x, y \in V(\pi_1) \cap V(\pi_2)$, where $x \neq y$, such that node $x$ comes before node $y$ in path $\pi_1$ and in path $\pi_2$. Then by the consistency of tiebreaking scheme $\pi(\cdot, \cdot)$, it follows that $\pi_1[x, y] = \pi_2[x, y]$. However, this implies that there exists an edge $e \in E(\pi_1) \cap E(\pi_2)$, contradicting Assumption 2 about the collection of paths $\pi(P)$. 
\end{proof}

Now, for any path $\pi^* \in \pi(P)$, we will define an associated collection of acyclic paths $\Pi_{\textsc{dag}}$. This will correspond to the acyclic collection of paths discussed in the technical overview in Observation \ref{obs:undirected}. 

\begin{definition}[DAG Path Collection]
\label{def:dag}
Fix a  path $\pi^* \in \pi(P)$. For every other path $\pi \in \pi(P) \setminus \{\pi^*\}$ we will define a subpath $\pi_{\textnormal{\textsc{dag}}}$ of $\pi$ induced on the nodes of $\pi^*$. This will yield an acyclic collection of paths $\Pi_{\textnormal{\textsc{dag}}} = \{\pi_{\textnormal{\textsc{dag}}} \mid \pi \in \pi(P) \setminus \{\pi^*\}\}$.  Given a path $\pi \in \pi(P) \setminus \{\pi^*\}$, we construct $\pi_{\textnormal{\textsc{dag}}}$ as follows.
\begin{itemize}
    \item Let $v_1, \dots, v_k \in V(\pi^*) \cap V(\pi)$ be the nodes in $ V(\pi^*) \cap V(\pi)$, listed in the order they appear in path $\pi$. 
    \item We define $\pi_{\textnormal{\textsc{dag}}}$ to be
    $$
    \pi_{\textsc{dag}} = (v_1, \dots, v_k).
    $$
    \item We assign edge $(v_i, v_{i+1}) \in \pi_{\textnormal{\textsc{dag}}}$ weight $w(v_i, v_{i+1}) = \dist_G(v_i, v_{i+1})$.  
\end{itemize}
 We define directed graph $D$ be the graph obtained by taking the (non-disjoint) union of paths in $\Pi_{\textnormal{\textsc{dag}}}$, i.e., $$D = \bigcup_{\pi_{\textnormal{\textsc{dag}}} \in \Pi_{\textnormal{\textsc{dag}}}} \pi_{\textnormal{\textsc{dag}}}.$$
\end{definition}
We now verify that $\Pi_{\textsc{dag}}$ is an acyclic collection of shortest paths in graph $D$. 
\begin{claim}
Fix a path $\pi^* \in \pi(P)$, and let $\Pi_{\textnormal{\textsc{dag}}}$ and $D$ be as defined in Definition \ref{def:dag}. Then graph $D$ is a DAG, and $\pi_{\textnormal{\textsc{dag}}}$ is a unique shortest path in $D$ for all $\pi_{\textnormal{\textsc{dag}}} \in \Pi_{\textnormal{\textsc{dag}}}$. 
\label{claim:usp_dag}
\end{claim}
\begin{proof}
    Let $<_D$ be the total ordering on the vertices of $V(D)$ that is the reverse of the ordering of the vertices in path $\pi^*$. That is, if node $x$ comes before node $y$ in $\pi^*$, then $y <_D x$. It is immediate that $<_D$ is a valid topological ordering for directed graph $D$ because every path $\pi_{\textsc{dag}} \in \Pi_{\textsc{dag}}$ respects $<_D$ by Claim \ref{clm:path_inter}.

    To see that all paths in $\Pi_{\textsc{dag}}$ are unique shortest paths in $D$,  observe the following properties of graph $D$:
    \begin{itemize}
        \item For all $s, t \in V(D)$,
            $$
            \dist_D(s, t) \geq \dist_G(s, t).
            $$
        \item For all vertices $s, t \in V(D)$, if $s, t \in V(\pi_{\textsc{dag}})$ for some $\pi \in \Pi_{\textsc{dag}}$,  then
        $$
        \dist_D(s, t) = \dist_G(s, t).
        $$
    \end{itemize}
These two properties imply that every path $\pi_{\textsc{dag}} \in \Pi_{\textsc{dag}}$ is a shortest path in $D$. To see that $\pi_{\textsc{dag}} $ is a \textit{unique} shortest path in $D$, observe that if $\pi_{\textsc{dag}}$ is an $s \leadsto t$ path in $D$, then path $\pi_{\textsc{dag}} $ is the unique (possibly non-contiguous) subpath of path $\pi(s, t) \subseteq G$ in $D$. Then since $\pi(s, t)$ is a unique shortest path in $G$, we conclude that $\pi_{\textsc{dag}}$ is a unique shortest path in $D$. 
\end{proof}

Our end goal is to upper bound the size $|E(H)|$ of graph $H$ by upper bounding the size $|E(D)|$ of DAG $D$. As such, we need to prove that $|E(D)|$ grows with $|E(H)|$. We achieve this in a roundabout way with the following claim.

\begin{claim}
\label{claim:degree}
    Fix a shortest path $\pi^* \in \pi(P)$, and let $\Pi_{\textnormal{\textsc{dag}}}$ and $D$ be as defined in Definition \ref{def:dag}. Let $Q \subseteq P$ be the set of demand pairs $(s, t)$ in $P$ such that $\pi(s, t)$ intersects path $\pi^*$ at exactly 1 node. Then
    $$
    \sum_{v \in V(\pi^*)} \deg_H(v) = \Theta(|\pi^*| + |Q| + |E(D)|).
    $$
\end{claim}
\begin{proof}
In order to prove Claim \ref{claim:degree}, it will suffice to tightly bound the number of distinct edges incident to nodes in $\pi^*$, since
$$
 \sum_{v \in V(\pi^*)} \deg_H(v) = \Theta(\left|  \{ e \in E(H) \mid e = (u, v), \text{ where $u \in \pi^*$ or $v \in \pi^*$} \}  \right|).
$$
Let $\Pi \subseteq \pi(P)$ be the set of all paths in $\pi(P)$ in graph $H$ that have a nonempty intersection with path $\pi^*$.
Observe that every edge incident to a node in $\pi^*$ in $H$ is contained in a path in $\Pi$. 
We partition the set of paths $\Pi$ into sets $\Pi_1, \Pi_2,$ and $\Pi_3$, where
$$
\Pi_1 = \{\pi^*\}, \quad \Pi_2 = \pi(Q), \quad \text{ and } \quad \Pi_3 = \Pi \setminus (\Pi_1 \cup \Pi_2). 
$$
We let $E_1, E_2,$ and $E_3$ denote the set of edges in paths in $\Pi_1, \Pi_2,$ and $\Pi_3$, respectively, that are incident to nodes in $\pi^*$. Then since sets $\Pi_1, \Pi_2$, and $\Pi_3$ partition $\Pi$, and paths in $\Pi$ are pairwise edge-disjoint by Assumption 2, we conclude that
$$
\sum_{v \in V(\pi^*)} \deg_H(v) = \Theta(|E_1|+ |E_2|+|E_3|).
$$
We will  bound the sizes of $E_1, E_2$, and $E_3$ separately. 
\begin{itemize}
    \item (Size of $|E_1|$.) Set $E_1$ has size $|E_1| = \Theta(|\pi^*|)$. 
    \item (Size of $|E_2|$.) Every path $\pi \in \pi(Q)$ has either 1 or  2 edges incident to a node in $\pi^*$.  
    Then since paths in $\pi(P)$ are pairwise edge-disjoint by Assumption 2,  the edges of the paths in $\pi(Q)$ contribute $\Theta(|Q|)$ edges to the sum of the degrees of the nodes in $\pi^*$ in $H$.
    \item (Size of $|E_3|$.) Fix a path $\pi \in \Pi_3$, and let $\pi_{\textsc{dag}}$ denote the corresponding path in $\Pi_{\textsc{dag}}$. We observe that since $|V(\pi) \cap V(\pi^*)| > 1 $,  
    $$
    |\{(u, v) \in E(\pi) \mid u \in V(\pi^*) \text{ or } v\in V(\pi^*) \}| = \Theta(|V(\pi) \cap V(\pi^*)|) = \Theta(|\pi_{\textsc{dag}}|). 
    $$
    Then we conclude that
    $$
    |E_3| = \sum_{\pi \in \Pi_3} |\{(u, v) \in E(\pi) \mid u \in V(\pi^*) \text{ or } v\in V(\pi^*) \}| = \Theta\left(\sum_{\pi \in \Pi_3} |\pi_{\textsc{dag}}| \right) = \Theta(|E(D)|),
    $$
    where the first and last equality follow from the fact that paths in $\Pi$ are pairwise edge-disjoint and that paths in  $\Pi_{\textsc{dag}}$ are pairwise edge-disjoint, respectively. 
\end{itemize}
Putting it all together, we conclude that
    $$
    \sum_{v \in V(\pi^*)} \deg_H(v) = 
 \Theta(|E_1| + |E_2|+|E_3|) = \Theta(|\pi^*| + |Q| + |E(D)|),
    $$
    as desired.
\end{proof}

We now have all the ingredients we need to prove our reduction from $\textsc{DDP}$ to $\textsc{DADP}$. 

\begin{lemma}[Directed to DAG Reduction]
$$
\textnormal{\DDP}(n, p) \leq O \left(\textnormal{\DADP}\left(\frac{\textnormal{\DDP}(n, p)}{p}, \hspace{2mm} p\right)^{1/2} \cdot n^{1/2}p^{1/2} + n^{1/2}p + n\right).
$$
\label{lem:dir_dag}
\end{lemma}
\begin{proof}
  Let $G, P$ be an $n$-node directed, weighted graph and  a set of $|P| = p$ demand pairs, such that the minimal distance preserver $H$ of $G, P$ has $$|E(H)| = \Theta(\textsc{DDP}(n, p))$$ edges. Let $\pi(\cdot, \cdot)$ be a tiebreaking scheme associated with $G, P$. We can assume that Assumptions 1-4 discussed at the beginning of Section \ref{sec:reduction} hold  without loss of generality.
  In particular, these assumptions imply that there exists a path $\pi^* \in \pi(P)$ such that $|\pi^*| = \Theta\left(\frac{|E(H)|}{p} \right)$, since the average length of paths in $\pi(P)$ is $\Theta\left(\frac{|E(H)|}{p} \right)$, and every path in $\pi(P)$ has length at least $\frac{|E(H)|}{4p}$.

  Let the collection of paths $\Pi_{\textsc{dag}}$ and DAG $D$ be as defined in Definition \ref{def:dag}, with respect to path $\pi^*$. By Assumption 3, we have that
  $$
 \sum_{v \in V(\pi^*)} \deg_H(v) = \Omega\left( \frac{|E(H)|}{p} \cdot \frac{|E(H)|}{n}\right) = \Omega\left( 
  \frac{\textsc{DDP}(n, p)}{np}^2 \right).
  $$
    Applying Claim \ref{claim:degree}, it follows that
    $$
    \Omega\left( 
  \frac{\textsc{DDP}(n, p)^2}{np} \right) \leq  \Theta(|\pi^*| + |Q| + |E(D)|),
    $$
    where $Q \subseteq P$ is the set of demand pairs $(s, t) \in P$ such that $\pi(s, t)$ intersects path $\pi^*$ at exactly 1 node. We make the following observations: 
    \begin{itemize}
        \item $|\pi^*| = \Theta\left( \frac{\textsc{DDP}(n, p)}{p}\right)$ by our choice of $\pi^*$,
        \item $|Q| \leq p$ since $Q \subseteq P$, and
        \item $|E(D)| \leq \textsc{DADP}(|V(D)|, |\Pi_{\textsc{dag}}|)$.  
    \end{itemize}
    The final observation follows from the fact that each path $\pi_{\textsc{dag}} \in \Pi_{\textsc{dag}}$ is a unique shortest path between its endpoints in $D$ by Claim \ref{claim:usp_dag}.  Consequently, if we define demand pairs 
    $$
    P_D = \{(s, t) \in V(D) \times V(D) \mid \text{ $\exists \text{ } \pi_{\textsc{dag}} \in \Pi_{\textsc{dag}}$ s.t.    $\pi_{\textsc{dag}}$ is an $s \leadsto t$ path in $D$} \},
    $$
    then any distance preserver of $D, P_D$ will contain every path in $\Pi_{\textsc{dag}}$ and thus $E(D)$. Using our observations and the fact that $|\Pi_{\textsc{dag}}| \leq p$ and $|V(D)| \leq |V(\pi^*)| = \Theta(\frac{\textsc{DDP}(n, p)}{p})$, we obtain the following inequality:
    $$
    \Omega\left( 
  \frac{\textsc{DDP}(n, p)^2}{np} \right) \leq  \Theta(|\pi^*| + |Q| + |E(D)|) \leq \Theta\left( \frac{\textsc{DDP}(n, p)}{p} + p + \textsc{DADP}\left(  \frac{\textsc{DDP}(n, p)}{p}, p \right) \right)  
    $$
    We now split our analysis into three cases, based on which of the three terms on the right hand side is the largest.
    \begin{itemize}
        \item $ \left( \text{Term }\frac{\textsc{DDP}(n, p)}{p} \text{ dominates.}\right)$ In this case, we find that
        $$
        \frac{\textsc{DDP}(n, p)^2}{np} \leq O\left( \frac{\textsc{DDP}(n, p)}{p} \right).
        $$
        Rearranging, we conclude that
        $$
        \textsc{DDP}(n, p) = O(n).
        $$
        \item $\left(\text{Term } p \text{ dominates.}\right)$ In this case, we find that
        $$
        \frac{\textsc{DDP}(n, p)^2}{np} \leq O\left(p \right).
        $$
        Then multiplying both sides of the inequality by $np$ and taking the square root of both sides, we get
        $$
        \textsc{DDP}(n, p) = O(n^{1/2}p).
        $$
        \item $\left(\text{Term } \textsc{DADP}\left(  \frac{\textsc{DDP}(n, p)}{p},  p \right)  \text{ dominates.}\right)$ In this case, our inequality is 
        $$
        \frac{\textsc{DDP}(n, p)^2}{np} \leq O\left( \textsc{DADP}\left(  \frac{\textsc{DDP}(n, p)}{p}, \quad p \right)\right).
        $$
        Multiplying both sides by $np$, and taking the square root of both sides, we obtain
        $$
        \textsc{DDP}(n, p) = O\left(\textsc{DADP}\left( \frac{\textsc{DDP}(n, p)}{p}, \quad p \right)^{1/2}\cdot n^{1/2}p^{1/2} \right)
        $$
    \end{itemize}
    Combining our three bounds from our three different cases, we conclude that
$$
\textsc{DDP}(n, p) \leq O \left(\textsc{DADP}\left(\frac{\textsc{DDP}(n, p)}{p}, \quad p\right)^{1/2} \cdot n^{1/2}p^{1/2} + n^{1/2}p + n\right),
$$
completing the proof.
\end{proof}

\subsection{DAG to Undirected Reduction}
\label{sec:dag_to_und}

We now present our reduction to convert shortest paths in DAGs to shortest paths in undirected graphs.

\begin{lemma}[DAG to Undirected Reduction]
    $$
    \textnormal{\DADP}(n, p) \leq \textnormal{\UDP}(n, p)
    $$
    \label{lem:dag_und}
\end{lemma}
\begin{proof}
       Let $G = (V, E, w)$ be an $n$-node weighted directed acyclic graph, and let $P \subseteq V(G) \times V(G)$ be a set of $|P|=p$ demand pairs. Let $<_G$ be a topological ordering of DAG $G$, and let
       $v_1, \dots, v_n$ be the nodes of $G$, ordered according to $<_G$. Let $W$ be a sufficiently large number (e.g., chose $W$ to be larger than the sum of all edge weights in $G$). We may assume without loss of generality that for all $(s, t) \in P$, there exists a unique shortest $s \leadsto t$ path $\pi(s, t)$ in $G$, by the Independence Lemma for weighted distance preservers (Lemma 33 of \cite{BHT22}).

       We define an undirected graph $G' = (V, E', w')$ as follows. For every directed edge $(v_i, v_j) \in E(G)$, add the undirected edge $(v_i, v_j)$ to $E(G')$, and assign it weight
       $$
       w'(v_i, v_j) := w(v_i, v_j) + W \cdot (j-i).
       $$
       This completes the construction of graph $G'$.

        Fix a pair of nodes $v_i, v_j \in V$  such that $v_i$ can reach $v_j$ in $G$. Let $\pi $ denote the unique shortest $v_i \leadsto v_j$  path in $G$, and let $\pi'$ denote a shortest $v_i \leadsto v_j$ path in $G'$. We claim that paths $\pi$ and $\pi'$ are identical, up to a reweighting of their edges. Formally, we claim if path $\pi$ has the corresponding node sequence $(x_1, x_2, \dots, x_k)$, then $\pi'$ has the corresponding node sequence  $(x_1, x_2, \dots, x_k)$, as well.

Suppose that path  $\pi$ has the corresponding node sequence $( x_1, x_2, \dots, x_k)$, where $x_1 = v_i$ and $x_k = v_j$. For all $\ell \in [1, k]$, let $q_{\ell} \in [1, n]$ be the index such that node $x_{\ell} = v_{q_{\ell}}$.  
Notice that since $\pi$ is a path in DAG $G$, we must have that $q_{\ell} < q_{\ell+1}$ for all $\ell \in [1, k-1]$. Then the distance between $v_i$ and $v_j$ in $G'$ is at most
\begin{align*}
\dist_{G'}(v_i, v_j) & \leq w'(\pi) = \sum_{i=1}^{k-1} w'(x_i, x_{i+1}) = \sum_{i=1}^{k-1} \left( w(x_i, x_{i+1}) + W \cdot |q_{i+1} - q_i|  \right) \\
& = \sum_{i=1}^{k-1}\left( w(x_i, x_{i+1}) + W \cdot (q_{i+1} - q_i)  \right) \\
& = w(\pi) + W \cdot(q_k - q_1) \\
& = \dist_G(s, t) +   W \cdot (j-i).
\end{align*}
In the above sequence of inequalities, we made use of the fact that since $q_{\ell} < q_{\ell + 1}$ for all $ \ell \in [1, k-1]$, it follows that
$$\sum_{i=1}^{k-1}|q_{i+1} - q_i| = \sum_{i=1}^{k-1}(q_{i+1} - q_i) = q_k - q_1 = j- i.$$

Let path $\pi'$ have the corresponding node sequence $(y_1, y_2, \dots, y_{k'})$, where $y_1 = v_i$ and $y_{k'} = v_j$.
For all $\ell \in [1, k']$, let $r_{\ell} \in [1, n]$ be the index such that node $y_{\ell} = v_{r_{\ell}}$. 
We claim that since $\pi'$ is a shortest $v_i \leadsto v_j$ path in $G'$, we have that $r_{\ell} < r_{\ell + 1}$ for all $\ell \in [1, k'-1]$. Informally speaking, this is equivalent to claiming that path $\pi'$ respects the total order $<_G$ in undirected graph $G'$.

Suppose towards contradiction that there exists an $\ell \in [1, k]$ such that $r_{\ell} > r_{\ell + 1}$. Then 
$$\sum_{i=1}^{k'-1}|r_{i+1} - r_i| > \sum_{i=1}^{k'-1}(q_{i+1} - q_i) = q_{k'} - q_1 = j - i.$$
This immediately implies that
\begin{align*}
w'(\pi') & = \sum_{i=1}^{k'-1}\left(w(y_i, y_{i+1}) + W \cdot |r_{i+1} - r_i| \right) \\
& \geq w(\pi') + W \cdot (j - i + 1) \\
& = \left(w(\pi') + W\right) + W \cdot (j-i) \\
& > w(\pi) + W\cdot (j-i) \\
& = \dist_{G'}(s, t),
\end{align*}
where the final inequality follows the fact that if $W$ is sufficiently large (e.g., larger than the sum of all edge weights in $G$), then $W > w(\pi)$. This contradicts our assumption that $\pi'$ is a shortest $v_i \leadsto v_j$ path in $G'$, so we conclude that $r_{\ell} < r_{\ell + 1}$ for all $\ell \in [1, k']$. 

Then by identical calculations as before, we have that
\begin{align*}
    w'(\pi') & = \sum_{i=1}^{k'-1}\left( w(y_i, y_{i+1}) + W \cdot (q_{i+1} - q_i)  \right) \\
& = w(\pi') + W \cdot(q_k - q_1) \\
& = w(\pi') +   W \cdot (j-i)
\end{align*}
Since $\pi'$ is a shortest path in $G'$, then
$$
w'(\pi') = w(\pi') + W \cdot (j - i)  \leq \dist_{G'}(v_i, v_j) \leq \dist_G(v_i, v_j) + W \cdot (j-i),
$$
so 
$
w(\pi') \leq \dist_G(v_i, v_j).
$ Since path $\pi$ is the unique shortest $v_i \leadsto v_j$  path in $G$, we conclude that paths $\pi$ and $\pi'$ are identical, e.g., $(x_1, \dots, x_k) = (y_1, \dots, y_{k'})$. 

Then every path $\pi \in \pi(P)$ is a unique shortest path between its endpoints in undirected graph $G'$. Informally, this means that the unique shortest paths between demand pairs $P$ are identical in $G$ and $G'$.  
As a consequence, any distance preserver of $G, P$ is a distance preserver of $G', P$, and vice versa. We conclude that 
$\textsc{DADP}(n, p) \leq \textsc{UDP}(n, p)$.
\end{proof}

Our  reduction from $\textsc{DADP}$ to $\textsc{UDP}$ also implies an algorithmic reduction from APSP in DAGs to APSP in undirected graphs, as we claimed in Section \ref{sec:intro:preserver} in Corollary \ref{corr:apsp}. We now prove Corollary \ref{corr:apsp}. 

\begin{corollary}[restatement of Corollary \ref{corr:apsp}]
\label{corr:apsp-main}
    Let $G$ be an $n$-node, $m$-edge weighted DAG. Then we can compute all-pairs shortest paths on $G$ in $O(m)$ time, plus the time it takes to compute all-pairs shortest paths on an $n$-node, $m$-edge undirected weighted graph.
\end{corollary}
\begin{proof}
First, we observe that given a directed acyclic graph $G$, we can construct graph $G'$ from Lemma \ref{lem:dag_und} in $O(n+m)$ time. Topological sorting can be done in $O(n+m)$ time. We can reweigh edge $(v_i, v_j)$ by adding $+W \cdot (j - i)$ weight to $(v_i, v_j)$ in $O (1)$ time.  Thus constructing undirected, weighted graph $G'$ can be done in $O(n+m)$ time. 
Let $v_1, \dots, v_n$ be the topological ordering of $G$ from Lemma \ref{lem:dag_und}.

For all $v_i, v_j \in V(G)$,  if
$$
\dist_{G'}(v_i, v_j) \geq W \cdot (j - i+1),
$$
then we report that $\dist_G(v_i, v_j) = \infty$.
Otherwise, we report that $$\dist_{G}(v_i, v_j) = \dist_{G'}(v_i, v_j) - W \cdot (j- i).$$ 
\end{proof}

\subsection{Finishing the Reduction}
\label{subsec:finishing_reduction}

We can chain our reductions in Lemmas \ref{lem:dir_dag} and \ref{lem:dag_und} to prove Theorem \ref{thm:reduction}. 

\begin{proof}[Proof of Theorem \ref{thm:reduction}]
    Suppose that $\textsc{UDP}(n, p) = O(n^{\lambda} p^{\mu} + n)$. Then we calculate
    \begin{align*}
\textsc{DDP}(n, p) & \leq O \left(\textsc{DADP}\left(\frac{\textsc{DDP}(n, p)}{p}, p\right)^{1/2} \cdot n^{1/2}p^{1/2} + n^{1/2}p + n\right) \\
 & \leq O \left(\textsc{UDP}\left(\frac{\textsc{DDP}(n, p)}{p}, p\right)^{1/2} \cdot n^{1/2}p^{1/2} + n^{1/2}p + n\right) \\
  & \leq O \left(\left(\left(\frac{\textsc{DDP}(n, p)}{p} \right)^{\lambda} p^{\mu} + \frac{\textsc{DDP}(n, p)}{p} \right)^{1/2} \cdot n^{1/2}p^{1/2} + n^{1/2}p + n\right) \\
 &  \leq O \left( \left( \frac{\textsc{DDP}(n, p)^{\lambda/2}}{p^{\lambda/2}} \cdot p^{\mu/2} + \frac{\textsc{DDP}(n, p)^{1/2}}{p^{1/2}} \right)  \cdot n^{1/2}p^{1/2} + n^{1/2}p + n\right) \\
  &  \leq O \left( \textsc{DDP}(n, p)^{\lambda/2} \cdot n^{1/2}p^{\frac{1  + \mu- \lambda}{2}} + \textsc{DDP}(n, p)^{1/2}\cdot n^{1/2} + n^{1/2}p + n 
  \right)
    \end{align*}
This upper bound has four terms. Suppose the first term $ \textsc{DDP}(n, p)^{\lambda/2} \cdot n^{1/2}p^{\frac{1  + \mu- \lambda}{2}}$ dominates. Then
\begin{align*}
    \textsc{DDP}(n, p)^{1 - \lambda/2} & \leq   O\left(n^{1/2}p^{\frac{1  + \mu- \lambda}{2}}\right),
\end{align*}
so in this case we conclude $$ \textsc{DDP}(n, p) \leq O\left(  n^{\frac{1}{2-\lambda}}p^{\frac{1 + \mu - \lambda}{2 - \lambda}}  \right).$$
Likewise, when the second term $\textsc{DDP}(n, p)^{1/2}\cdot n^{1/2}$ dominates, we get that $\textsc{DDP}(n, p)  = O(n)$. Putting it all together, we conclude that
$$
\textsc{DDP}(n, p) = O\left(n^{\frac{1}{2-\lambda}}p^{\frac{1 + \mu - \lambda}{2 - \lambda}} + n^{1/2}p + n \right). 
$$
\end{proof}

We can use a similar analysis to prove Corollary \ref{corr:reduction}. We   restate Corollary \ref{corr:reduction} in greater detail below for convenience.
\begin{corollary}[cf. Corollary \ref{corr:reduction}]
\label{corr:reduction-main}
Suppose that consistent tiebreaking schemes are optimal for directed distance preservers when $p \leq n^{2/3}$, or equivalently, 
$$
\textnormal{\DDP}(n, p) = \Theta(\min(n^{2/3}p + n, np^{1/2})) = \Theta(n^{2/3}p + n),
$$
when $p \leq n^{2/3}$. 
Then consistent tiebreaking schemes are optimal for undirected distance preservers when $p \leq n$, or equivalently, 
$$
\textnormal{\UDP}(n, p) = \Theta(\min(n^{1/2}p + n, np^{1/2})) = \Theta(n^{1/2}p + n),
$$
when $p \leq n$. 
\end{corollary}
\begin{proof}
    Suppose that $$
\textsc{DDP}(n, p) = \Theta(\min(n^{2/3}p + n, np^{1/2})) =  \Theta(n^{2/3}p + n) = \Theta(n^{2/3}p),
$$
when $n^{1/3} \leq p \leq n^{2/3}$. 
Now by combining Lemmas \ref{lem:dir_dag} and \ref{lem:dag_und}, we get the inequality
$$
\textsc{DDP}(n, p) \leq  O \left(\textsc{UDP}\left(\frac{\textsc{DDP}(n, p)}{p}, p\right)^{1/2} \cdot n^{1/2}p^{1/2} + n^{1/2}p + n\right).
$$
Plugging in $\textsc{DDP}(n, p) = \Theta(n^{2/3}p)$, we compute that
$$
n^{2/3}p \leq  O \left(\textsc{UDP}\left(n^{2/3}, p\right)^{1/2} \cdot n^{1/2}p^{1/2} + n^{1/2}p + n\right).
$$
Let $c> 0$ be a universal constant such that
$$
n^{2/3}p \leq  c \cdot \left(\textsc{UDP}\left(n^{2/3}, p\right)^{1/2} \cdot n^{1/2}p^{1/2} + n^{1/2}p + n\right),
$$
for sufficiently large $n$. Then for all $10cn^{1/3} \leq p \leq n^{2/3}$,
$$
n^{2/3} p \leq 10c \cdot \textsc{UDP}\left(n^{2/3}, p\right)^{1/2} \cdot n^{1/2}p^{1/2}.
$$
Dividing both sides by $10cn^{1/2}p^{1/2}$ and taking the square root, we obtain
$$
\textsc{UDP}(n^{2/3}, p) = \Omega(n^{1/3}p).
$$
If we let $N = n^{2/3}$, this is equivalent to
$$
\textsc{UDP}(N, p) = \Omega(N^{1/2}p),
$$
for $10c \cdot N^{1/2} \leq p \leq N$. Moreover, it is true unconditionally that $\textsc{UDP}(N, p) = \Theta(N)$ when $p \leq N^{1/2}$. We conclude that
$$
\textsc{UDP}(N, p) = \Omega(N^{1/2}p + N),
$$
for $p \leq N$. Consistent tiebreaking schemes imply that
$$
\textsc{UDP}(N, p) = O(\min(N^{1/2}p + N, Np^{1/2})) = O(N^{1/2}p + N),
$$
when $p \leq N$, so we conclude that 
$$
\textsc{UDP}(N, p) = \Theta(N^{1/2}p + N) = \Theta(\min(N^{1/2}p + N, Np^{1/2})) ,
$$
when $p \leq N$, as desired.
\end{proof}

\section{Exact Distance Preserver Upper Bound}
\label{sec:pres_up}

The goal of this section will be to prove the following theorem. 

\begin{theorem}
Every $n$-node directed unweighted graph $G$ with demand pairs $P$ of size $p = |P|$ admits a distance preserver with
$$
\widetilde{O}\left(n^{5/6}p^{2/3} + n \right)$$ 
edges.
\label{thm:unw_dp}
\end{theorem}

We will need to introduce some preliminary definitions first.

\subsection{Preliminary Definitions}

Throughout Section \ref{sec:pres_up}, we will use $G$ to denote an $n$-node directed unweighted graph, and we will use $P \subseteq V(G) \times V(G)$ to denote an associated set of demand pairs of size $|P| = p$. Additionally, we will associate with $G$ and  $P$ a \textit{consistent} tiebreaking scheme $\pi(\cdot, \cdot)$. Recall that for each demand pair $(s, t) \in P$,  our tiebreaking scheme fixes an $s \leadsto t$ shortest path $\pi(s, t)$. We let $\pi(P)$  denote this collection of shortest paths in $G$, i.e., 
$$
\pi(P) = \{\pi(s, t) \mid (s, t) \in P\}.
$$
Throughout this section, we will assume that every edge $e \in E(G)$ is contained in a path in $\pi(P)$. (This assumption is without loss of generality because we can safely remove all other edges in $G$ without increasing distances between demand pairs.)

Given graph $G$ and demand pairs $P$, we define the following quantities. 
\begin{itemize}
    \item We define quantity $\ell = \ell(G, P)$ to be 
    $$
    \ell(G, P) := \left \lceil \frac{|E(G)|}{p} \right \rceil.
    $$
    Quantity $\ell$ roughly corresponds to the average number of edges each path in $\pi(P)$  contributes to $|E(G)|$.
    \item We define quantity $d = d(G, P)$ to be
    $$
    d(G, P) := \left \lceil \frac{|E(G)|}{n} \right \rceil.
    $$
    Quantity $d$ roughly corresponds to the average degree of nodes in $G$. 
\end{itemize}
We will assume without loss of generality that $\ell$ and $d$ are sufficiently large. If $\ell \leq c$ or $d \leq c$ for some constant $c$, then $|E(G)| \leq cp$ or $|E(G)| \leq cn$, respectively. Thus our assumption that $\ell > c$ and $d > c$ for a sufficiently large constant $c$ will add $$O(\ell p + dn) = O(p + n)$$
additional edges to our final preserver $H$. With this assumption, it follows that
$$
\ell p = \Theta(|E(G)|) = \Theta(dn). 
$$
We will use the equality repeatedly in our proof. We now define an analogous quantity $\widehat{\ell}$  that captures how much paths in $\pi(P)$ overlap. 
\begin{itemize}
    \item We define $\widehat{\ell} = \widehat{\ell}(G, P)$ to be
    $$
    \widehat{\ell} := \frac{\sum_{\pi \in \pi(P)}|\pi|}{p}.
    $$
    Quantity $\widehat{\ell}$ roughly corresponds to the average length of the paths in $\pi(P)$. Note that $\widehat{\ell} \gg \ell$ in general, and $\widehat{\ell} = \Theta(\ell)$  when paths in $\pi(P)$ are pairwise edge-disjoint. 
\end{itemize}

We will make the following assumption about paths in $\pi(P)$. For every path $\pi = \pi(s, t)$, we have that
$$
\widehat{\ell} / 2 \leq |\pi| \leq 2\widehat{\ell}. 
$$
This assumption holds without loss of generality by a standard bucketing argument. That is, we can bucket the demand pairs $(s, t) \in P$ into buckets $P_0, \dots, P_{\log n}$ so that 
$$
P_i = \{(s, t) \in P \mid 2^i \leq  \dist_G(s, t) \leq 2^{i+1} \}
$$
for $i \in [0, \log n]$. Then we can apply our bound on $G, P_i$ for each $i  \in [0, \log n]$. This will worsen our upper bound by at most a logarithmic factor.

\subsection{Ingredient 1: the Dense Low-Diameter Cluster Lemma}
The first ingredient in our distance preserver construction will be the following lemma, which roughly states that if our distance preserver is sufficiently dense, then it contains a low-diameter subset of vertices containing many high-degree nodes.
We defer the proof of this lemma to Section \ref{subsec:low_diam_scc}.

\begin{lemma}[Dense Low-Diameter Cluster Lemma]
Let $G$ be an $n$-node directed unweighted graph with associated set of demand pairs $P \subseteq V(G) \times V(G)$ of size $|P| = p$.\footnote{We assume without loss of generality that every edge in $E(G)$ lies on a path in $\pi(P)$ for a consistent tiebreaking scheme $\pi(\cdot, \cdot)$.} Let $\ell = \ell(G, P)$, $d = d(G, P)$, and $\widehat{\ell} = \widehat{\ell}(G, P)$. If $G$ has average degree $d$, where $d \geq 10p/n^{1/2}$, then $G$ contains a collection of nodes $S \subseteq V(G)$ such that
\begin{enumerate}
    \item $|S| = \Theta(d)$, 
    \item{(Dense.)} For all $s \in S$, $\deg_G(s) = \Omega(d)$, and
    \item{(Low diameter.)} For all $s, t \in S$, $\dist_G(s, t) = O\left(\frac{\widehat{\ell}p^2}{nd^2 }\right)$.
\end{enumerate}
\label{lem:low-diam-scc}
\end{lemma}

An important feature of the set of vertices $S$ identified in Lemma \ref{lem:low-diam-scc} is that $S$ contains $\Theta(d)$ nodes each of degree $\Omega(d)$ in $G$. We will see how this property implies that many paths in $\pi(P)$ pass through nodes in $S$, as we state in the following lemma.

\begin{lemma}
    \label{lem:many-paths}
    Let $G, P, S$ and $\ell, d, \widehat{\ell}$ be as described in Lemma \ref{lem:low-diam-scc}.
    Let $Q \subseteq P$ be the set of demand pairs $(s, t) \in P$ such that $\pi(s, t)$ contains a node in $S$. Then 
    $$|Q| = \Omega(d^2).$$
\end{lemma}

Before we can prove  Lemma \ref{lem:many-paths}, we will first need to introduce some new notation and establish several technical claims. Let $\Pi \subseteq \pi(P)$ be a  collection of paths in $\pi(P)$, and let $S \subseteq V(G)$ be a set of nodes in $G$. 
Let $\Pi[S]$ denote the collection of paths in $\Pi$ induced on the nodes in set $S$. Formally, for a path $\pi \in \Pi$, we define the path $\pi[S]$ induced on the set $S$ to be the path obtained by deleting every node in $\pi \setminus S$ from $\pi$. Likewise, we let $\Pi[S] = \{\pi[S] \mid \pi \in \Pi\}$. Note that the paths in $\Pi[S]$ do not necessarily correspond to paths in graph $G$, as $E(\pi[S]) \not \subseteq E(\pi)$, in general. 

We quickly observe that if a collection of paths $\Pi$ is consistent, then the collection of paths $\Pi[S]$ is also consistent. 
\begin{claim}
    \label{clm:consistent}
    Let $\Pi$ be a consistent collection of paths, and let $S \subseteq V(G)$ be a set of nodes. Then the collection of paths $\Pi[S]$ is also consistent. 
\end{claim}

We  need an additional definition and technical claim before we can prove Lemma \ref{lem:many-paths}. Given a collection of paths $\Pi$ and a set of nodes $S \subseteq V(G)$, we define an auxiliary graph $J(S, \Pi)$ and an associated collection of paths $\Pi_J$ as follows. 

\begin{definition}[Graph $J$ and paths $\Pi_J$]
    Let $\Pi$ be a collection of paths, and let $S \subseteq V(G)$ be a set of nodes. Then we define a graph $J = J(S, \Pi)$ and associated collection of paths $\Pi_J$ as follows. 
    \begin{itemize}
        \item Let $V(J) = S$,
        \item Let $\Pi_J = \Pi[S]$, and
        \item Let the edge set $E(J)$ of $J$ be
        $$
        E(J) = \bigcup_{\pi \in \Pi_J} E(\pi).
        $$
    \end{itemize}
    \label{def:graph_J}
\end{definition}

The definition of graph $J$ and paths $\Pi_J$ we be useful in the proof of Lemma \ref{lem:many-paths}, when combined with the following technical claim.

\begin{claim}
   Let $G, P, S$ and $\ell, d, \widehat{\ell}$ be as described in Lemma \ref{lem:low-diam-scc}. Let $Q \subseteq P$ be the set of demand pairs $(s, t) \in P$ such that $\pi(s, t)$ contains a node in $S$.  Let $J = J(S, \pi(Q))$ and $\Pi_J = \pi(Q)[S]$ be the graph and collection of paths specified in Definition \ref{def:graph_J} with respect to set $S$ and paths $\pi(Q)$. Then
    $$
    |V(J)| = |S| = \Theta(d) \quad \text{ and } \quad |E(J)| = \Omega(d^2) - |Q|.
    $$
    \label{clm:induced_paths}
\end{claim}
\begin{proof}
It is immediate from Definition \ref{def:graph_J} that $|V(J)| = |S| = \Theta(d)$. What remains is to prove that $|E(J)| = \Omega(d^2) - |Q|$.  

Let $E' \subseteq E(G)$ be the set of edges in $G$ incident to a node in $S$. Formally, let $$
E' = \{(u, v) \in E(G) \mid u \in S \text{ or } v \in S\}.
$$ 
By Lemma \ref{lem:low-diam-scc}, $|E'| = \Omega(d) \cdot |S| = \Omega(d^2)$. We can assume without loss of generality that every edge $e \in E'$ lies on a path $\pi$ in $\pi(Q)$. We will define a subset $E'' \subseteq E'$ of $E'$ as follows. 
$$
E'' = \{e \in E' \mid \text{there exists $\pi \in \pi(Q)$ and $s, t \in V(\pi) \cap S$ such that $e \in \pi[s \leadsto t]$} \}.
$$
Informally, $E''$ corresponds to the set of edges in $E'$ that lie between two nodes $s, t \in S$ on a path $\pi \in \pi(Q)$. Notice that for every path $\pi \in \pi(Q)$, among the set of edges $E(\pi) \cap E'$, only the first  and  last edge in $E(\pi) \cap E'$ to appear on path $\pi$ are not contained in $E''$. Consequently,
$$
|E(\pi) \cap E''| \geq |E(\pi) \cap E'| - 2.
$$
This immediately gives the bound
$$
|E''| \geq |E'| - 2|Q| = \Omega(d^2) - 2|Q|. 
$$

To prove Claim \ref{clm:induced_paths}, we want to translate our lower bound on $|E''|$ into a lower bound on $|E(J)|$. Towards this end, we will define a mapping $\varphi:E'' \mapsto E(J)$ as follows. Let $e$ be an edge in $E''$. Since $e \in E''$, there exists a path $\pi \in \pi(Q)$ and nodes $s, t \in V(\pi) \cap S$ such that $e \in \pi[s \leadsto t]$. Moreover, we can assume without loss of generality that $(s, t) \in E(J)$, by choosing nodes $s$ and $t$ that are closest to edge $e$ on path $\pi$. We let $\varphi(e) = (s, t)$. We can define $\varphi(e)$ in this way for all $e \in E''$.

We claim that for every edge $(s, t) \in E(J)$, the number of edges in the preimage of edge $(s, t)$ is at most two, i.e.,   $\left|\varphi^{-1}((s, t))\right| \leq 2$. Fix an edge $(s, t) \in E(J)$. Let $E''_{(s, t)} \subseteq E''$ denote the set of edges $e$ in $E''$ such that $\varphi(e) = (s, t)$. If $|E''_{(s, t)}| = 0$, then we are done. Otherwise, there must exist a path $\pi^* \in \pi(P)$ and nodes $s, t \in V(\pi^*) \cap S$ such that node $s$ precedes node $t$ in path $\pi^*$. Since the collection of paths $\pi(P)$ is consistent, for every path $\pi \in \pi(P)$ such that $s, t \in V(\pi)$ and node $s$ precedes node $t$, we have that $\pi[s \leadsto t] = \pi^*[s \leadsto t]$. Additionally, since $(s, t) \in E(J)$, path $\pi^*[s \leadsto t]$ is internally vertex-disjoint from set $S$. Consequently, only the first and last edge of path $\pi^*[s \leadsto t]$ are contained in $E''$. This implies that $|E''_{(s, t)}| \leq 2$, as claimed. 

Putting everything together, we find that
$$
|E(J)| \geq \frac{1}{2} \cdot \left|\varphi^{-1}(E(J))\right| = \frac{1}{2} \cdot |E''| = \Omega(d^2) - |Q|, 
$$
as claimed. 
\end{proof}

With Definition \ref{def:graph_J} and  Claim \ref{clm:induced_paths} in hand, we are now ready to prove Lemma \ref{lem:many-paths}. 

\begin{proof}[Proof of Lemma \ref{lem:many-paths}]
    Let $G, P, S$ and $\ell, d, \widehat{\ell}$ be as described in Lemma \ref{lem:low-diam-scc}. Let $Q \subseteq P$ be the set of demand pairs $(s, t) \in P$ such that $\pi(s, t)$ contains a node in $S$.  Let $J = J(S, \pi(Q))$ and $\Pi_J = \pi(Q)[S]$ be the graph and collection of paths specified in Definition \ref{def:graph_J} with respect to set $S$ and paths $\pi(Q)$.

    Note that by Claim \ref{clm:consistent}, the collection of paths $\Pi_J$ is consistent. Additionally, by Claim \ref{clm:induced_paths} $|V(J)| = \Theta(d)$ and $|E(J)| \geq \Omega(d^2) - |Q|$. Now if $|Q| \geq \Omega(d^2)$, then we have finished the proof of Lemma \ref{lem:many-paths}. Otherwise, we may assume that $|E(J)| = \Omega(d^2)$. 

    Then $\Pi_J$ is a consistent collection of $|\Pi_J| = |Q|$ paths, over a vertex set $S$ of size $|S| = \Theta(d)$. Moreover, the number of distinct edges in $\Pi_J$ is at least $|E(J)| = \Omega(d^2)$. By Theorem 1.1 of \cite{CE06}, a consistent tiebreaking scheme for $x$ demand pairs over $y$ nodes has at most $O(yx^{1/2})$ distinct edges. Plugging in $x \geq |Q|$ and $y = |S| =\Theta(d)$, this implies that graph $J$ has at most $|E(J)| = O(d|Q|^{1/2}|)$ distinct edges. Putting everything together,
    $$
    \Omega(d^2) = |E(J)| = O(d|Q|^{1/2}),
    $$
    so $|Q| = \Omega(d^2)$, as claimed.
\end{proof}

\subsection{Ingredient 2: Sourcewise Distance Preservers}

The second ingredient in our distance preserver construction will be a sourcewise distance preserver upper bound for a collection of sources $S$ with low diameter. This sourcewise preserver is a direct generalization of the sourcewise preserver of \cite{BV21} to the directed setting, and follows from a similar argument.

\begin{lemma}[cf. Lemma 8 of \cite{BV21}]
Let $G$ be an $n$-node directed graph, and let $S \subseteq V(G)$ be a set of $|S| = s$ nodes of weak diameter $h$ (i.e., $\dist_G(s, t) \leq h$ for all $s, t \in S$). Let $P \subseteq S \times V$ be a set of $|P| = p$ sourcewise demand pairs in $G$. Then there exists a distance preserver of $G, P$ of size
$$
O\left((nsph)^{1/2} + n \right).
$$
    \label{lem:sourcewise_pres}
\end{lemma}
\begin{proof}
We defer the proof of Lemma \ref{lem:sourcewise_pres} to Appendix \ref{app:swise} due to its similarities with the proof of Lemma 8 in \cite{BV21}.
\end{proof}

\subsection{Main Construction}

In the first part of our main construction, we will  use Lemmas \ref{lem:low-diam-scc} and \ref{lem:sourcewise_pres} to construct an exact distance preserver $H$ that is sparse when $\widehat{\ell}$ is small.

\paragraph{Construction of distance preserver $H$ for small $\widehat{\ell}$.} Let $\pi(\cdot, \cdot)$ be a consistent tiebreaking scheme of $G, P$. 
We may assume without loss of generality that every edge in $E(G)$ is contained in a path in $\pi(P)$. 
Let $\varphi$ be a parameter of the construction that we will optimize later.  Our choice of $\varphi$ will satisfy $\varphi \geq 10p/n^{1/2}$. 
While the average degree of $G$ is at least $\varphi$, we will repeat the following process.

By Lemma \ref{lem:low-diam-scc} and Lemma \ref{lem:many-paths}, there exists a set of nodes $S \subseteq V(G)$ of size $|S| = \Theta(\varphi)$ such that:
\begin{itemize}
    \item The set $Q \subseteq P$ of demand pairs  $(s, t)$ such that $\pi(s, t)$ contains a node in $S$ is of size $|Q| = \Omega(\varphi^2)$. 
    \item The set $S$ has weak diameter  $O(\widehat{\ell}p^2/(n\varphi^2))$ in $G$.  
\end{itemize}

We will handle $\Theta(\varphi^2)$ paths in $Q$ using two instances of the sourcewise distance preserver claimed in Lemma \ref{lem:sourcewise_pres}. For each demand pair $(x, y) \in Q$ such that $\pi(x, y)$ contains node $s \in S$, we split demand pair $(x, y)$ into the two demand pairs $(x, s)$ and $(s, y)$. We let $P_1 \subseteq S \times V$ and $P_2 \subseteq V \times S$ be the two resulting sets of demand pairs. Note that in order to construct a distance preserver of $G, Q$, it suffices to take the union of a distance preserver of $G, P_1$ and a distance preserver of $G, P_2$. We can apply Lemma \ref{lem:sourcewise_pres} with parameters $s_1 := |S| = \Theta(\varphi)$, $p_1:= \min(|Q|, \varphi^2)=\Theta(\varphi^2)$,  and $h_1 := O(\widehat{\ell}p^2/(n\varphi^2))$ to construct a preserver $H_1$ of demand pairs $Q$ of size
$$
O\left((ns_1p_1h_1)^{1/2}+n \right) = O\left(
n^{1/2}\varphi^{3/2} \cdot \frac{\widehat{\ell}^{1/2}p}{n^{1/2}\varphi} + n
\right) = O\left(\widehat{\ell}^{1/2}\varphi^{1/2}p + n\right).
$$

We add the edges of $H_1$ to our final distance preserver $H$ of $G, P$. Since we have handled all demand pairs in $Q$, we can delete $Q$ from $P$ and repeat the above process. Formally, we update $G, P$ as follows:
\begin{itemize}
    \item $P \leftarrow P \setminus Q$, and
    \item $E(G) \leftarrow \bigcup_{(s, t) \in P}E(\pi(s, t))$.
\end{itemize}
We continue this process until the average degree of $G$ is less than $\varphi$. 

Since $p_1=\Theta(\varphi^2)$, we will repeat this process at most $p/p_1=O(p/\varphi^2)$ times before the average degree of $H_1$ is less than $\varphi$.
This phase of the construction contributes at most
$$
\frac{p}{p_1} \cdot O\left(\widehat{\ell}^{1/2}\varphi^{1/2}p + n\right) = 
O\left(
    \widehat{\ell}^{1/2}\cdot \frac{p^2}{\varphi^{3/2}} + \frac{np}{\varphi^2} \right)
$$
edges to $H$. 
We finish the construction of $H$ by adding all remaining $O(n\varphi)$ edges in $H_1$ to $H$. This yields the following lemma.
\begin{lemma}
For every $\varphi \geq 10p/n^{1/2}$, there exists a distance preserver $H$ of $G, P$ with size
    $$
    |E(H)| = \widetilde{O}\left(
    \widehat{\ell}^{1/2}\cdot \frac{p^2}{\varphi^{3/2}} + \frac{np}{\varphi^2} +n\varphi \right).
    $$
    \label{lem:small_hell}
\end{lemma}

To complement the distance preserver of Lemma \ref{lem:small_hell}, we will construct a distance preserver that is sparse when $\widehat{\ell}$ is large.

\paragraph{Distance preserver $H$ for large $\widehat{\ell}$.} 

To preserve distances when $\widehat{\ell}$ is large, we will make use of a theorem of \cite{BCE05}, which states that we can preserve distances between vertices that are far apart in $H$ using few edges.
\begin{theorem}[Theorem 2.24 of \cite{BCE05}]
    \label{thm:big_hell}
    Given a (possibly directed) unweighted graph $G$, a $D$-preserver is a subgraph $H$ of $G$ such that for all $s, t \in V(G)$ with $\dist_G(s, t) \geq D$, we have that $$\dist_H(s, t) = \dist_G(s, t).$$
    There exists a $D$-preserver $H$ of $G$ with  size 
    $
    O(n^2/ \widehat{\ell}).
    $
\end{theorem}

As an immediate consequence of Theorem \ref{thm:big_hell}, we can obtain a distance preserver $H$ of $G, P$ with size $O(n^2/\widehat{\ell})$, since we can assume all shortest paths $\pi$ in $\pi(P)$ are of length at least $|\pi| \geq \widehat{\ell}/2$ without loss of generality, by our earlier bucketing argument.

\begin{lemma}
    \label{lem:big_hell}
    There exists a distance preserver $H$ of $G, P$ with  size
    $
    O(n^2/ \widehat{\ell}).
    $
\end{lemma}

\paragraph{Finishing the proof of Theorem \ref{thm:unw_dp}.}

\begin{claim}
    \label{clm:pres_up}
    Together, Lemmas \ref{lem:small_hell} and \ref{lem:big_hell} imply a distance preserver $H$ of $G, P$ of size 
    $$
    |E(H)| = \widetilde{O}\left(  \frac{n^{2/3}p^{4/3}}{\varphi} + \frac{np}{\varphi^2} + n\varphi \right).
    $$
\end{claim}
\begin{proof}
We will use the construction from Lemma \ref{lem:small_hell} when $\widehat{\ell}$ is small, and we will use the construction from Lemma \ref{lem:big_hell} when $\widehat{\ell}$ is large. Specifically, we define our threshold $t$ for $\widehat{\ell}$ to be $$t = \frac{n^{4/3}\varphi}{p^{4/3}}.$$
We calculate the following:
\begin{itemize}
    \item When $\widehat{\ell} \leq t$,
    $$
     \widehat{\ell}^{1/2}\cdot \frac{p^2}{\varphi^{3/2}} \leq t^{1/2} \cdot \frac{p^2}{\varphi^{3/2}} = O\left(\frac{n^{2/3}p^{4/3}}{\varphi} \right),
    $$
    so the distance preserver $H$ of Lemma \ref{lem:small_hell} has size
    $$
    |E(H)| = \widetilde{O}\left(\frac{n^{2/3}p^{4/3}}{\varphi} + \frac{np}{\varphi^2} + n\varphi \right),
    $$
    as desired.
    \item When $\widehat{\ell} > t$, distance preserver $H$ of Lemma \ref{lem:big_hell} has size
    $$
    |E(H)| = \widetilde{O}(n^2/\widehat{\ell}) \leq   \widetilde{O}(n^2/t) = \widetilde{O}\left( \frac{n^{2/3}p^{4/3}}{\varphi}\right),
    $$
    as desired. \qedhere
\end{itemize}
\end{proof}

We are now ready to prove Theorem \ref{thm:unw_dp}.
\begin{proof}[Proof of Theorem \ref{thm:unw_dp}]
The bounds stated in Theorem \ref{thm:unw_dp} will follow directly from the bounds achieved in Claim \ref{clm:pres_up} after parameter balancing. Initially, we will assume that $n^{1/2} \leq p \leq n$. 

Straightforward calculations show that the three terms in Claim \ref{clm:pres_up} are balanced when 
$$
\varphi = \frac{p^{2/3}}{n^{1/6}}.
$$
Recall that parameter $\varphi$ must be at least 
$
\varphi \geq 10p/n^{1/2}
$
for Lemma \ref{lem:low-diam-scc} to go through. This holds for  our choice of $\varphi = \frac{p^{2/3}}{n^{1/6}}$ when $p \leq n$. 
Plugging this value of $\varphi$ into Claim \ref{clm:pres_up}, we obtain the upper bound
$$
\widetilde{O}\left( n^{5/6}p^{2/3} + \frac{n^{4/3}}{p^{1/3}} + n^{5/6}p^{2/3} \right).
$$
We observe that $\frac{n^{4/3}}{p^{1/3}} \leq n^{5/6}p^{2/3}$ when $p \geq n^{1/2}$. We conclude that
$
|E(H)| = \widetilde{O}( n^{5/6}p^{2/3})
$ when $n^{1/2} \leq p \leq n$. 
For the settings where $p < n^{1/2}$ or $p > n$, we observe that the existing upper bounds of $\min(n^{2/3}p+n, np^{1/2})$ for directed distance preservers due to \cite{CE06, Bodwin21} are smaller than our claimed bound of $\widetilde{O}(n^{5/6}p^{2/3} + n)$.
\end{proof}

\subsection{Proving the Dense Low-Diameter Cluster Lemma}
\label{subsec:low_diam_scc}
The objective of this section is to prove Lemma \ref{lem:low-diam-scc}.

\begin{lemma*}[Dense Low-Diameter Cluster Lemma]
Let $G$ be an $n$-node directed unweighted graph with associated set of demand pairs $P \subseteq V(G) \times V(G)$ of size $|P| = p$. Let $\ell = \ell(G, P)$, $d = d(G, P)$, and $\widehat{\ell} = \widehat{\ell}(G, P)$. If $G$ has average degree $d$, where $d \geq 10p/n^{1/2}$, then $G$ contains a collection of nodes $S \subseteq V(G)$ such that
\begin{enumerate}
    \item $|S| = \Theta(d)$, 
    \item{(Dense.)} For all $s \in S$, $\deg_G(s) = \Omega(d)$, and
    \item{(Low diameter.)} For all $s, t \in S$, $\dist_G(s, t) = O\left(\frac{\widehat{\ell}p^2}{nd^2 }\right)$.
\end{enumerate}
\end{lemma*}

Before proving this lemma, we will need to introduce several new definitions, which will be useful for finding our dense low-diameter cluster. Let $G, P, S$ and $\ell, d, \widehat{\ell}$ be as in the statement of Lemma \ref{lem:low-diam-scc}.

Recall the definition of a branching event from Definition \ref{def:be}. 
 We say that a branching event $(x, y), (x, z)$ is a \textit{high-degree branching event} if $\text{outdeg}_G(x) \geq d/4$. A simple argument will show that $G$ contains many high-degree branching events. 
 \begin{claim}[\cite{CE06}]
     $G$ contains $\Omega(nd^2)$ high-degree branching events.
     \label{clm:be}
 \end{claim}
 \begin{proof}
Since $G$ has average degree $d$, it follows that
$$
\sum_{v \in V(G)} \text{\normalfont outdegree}_G(v) = nd/2.
$$
In particular, we have that 
\begin{align*}
& \sum_{\{v \in V(G) \mid \outdeg_G(v) \geq d/4\}} \text{\normalfont outdegree}_G(v)  \\
& =  \sum_{v \in V(G)} \text{\normalfont outdegree}_G(v) \quad -  \sum_{\{v \in V(G) \mid \outdeg_G(v) < d/4\}} \text{\normalfont outdegree}_G(v) \\
& \geq nd/2 - nd/4 \geq nd/4.
\end{align*}
Then the number of high-degree branching events is at least
\begin{align*}
    \sum_{\{v \in V(G) \mid \outdeg_G(v) \geq d/4\}} \text{\normalfont outdegree}^2_G(v) & \\
    & \geq \frac{1}{n} \cdot \left(  \sum_{\{v \in V(G) \mid \outdeg_G(v) \geq d/4\}} \text{\normalfont outdegree}_G(v) \right)^2 \\
    & \geq \frac{1}{n}\cdot \frac{n^2d^2}{16} = \Omega(nd^2), 
\end{align*}
where the second to last inequality follows from Cauchy-Schwarz. 
 \end{proof}
We will also introduce an auxiliary graph, called the path intersection graph, which will be useful in our proof of Lemma \ref{lem:base_path}.

\begin{definition}[Path Intersection Graph $J$]
    We define the vertex set of $J$ to be 
    $$
    V(J) = \pi(P),
    $$
    the set of all shortest paths associated with our tiebreaking scheme $\pi(\cdot, \cdot)$. For every $\pi_1, \pi_2 \in V(J)$, we add the edge $(\pi_1, \pi_2)$ to $J$ and we assign this edge weight equal to the number of high-degree branching events between $\pi_1$ and $\pi_2$. 
\end{definition}

\begin{lemma}
    There exists a path $\pi^* \in \pi(P)$ and a collection of paths $\pi(Q)$, where $Q \subseteq P$, satisfying the following properties:
    \begin{itemize}
        \item Every path $\pi' \in \pi(Q)$ has at least $\Omega(nd^2/p^2)$ high-degree branching events with $\pi^*$.
        \item Let $S \subseteq V(\pi^*)$ be the set of nodes in $\pi^*$ with outdegree at least $d/4$, i.e., $$S = \{v \in V(\pi^*) \mid \outdeg_G(v) \geq d/4\}.$$ 
        The total number of high-degree branching events between path $\pi^*$ and the paths in $\pi(Q)$ is at least $\Omega(|S| d)$.
    \end{itemize}
    \label{lem:base_path}
\end{lemma}
\begin{proof}
Let $b$ denote the number of high-degree branching events in $G$.  
    By Claim \ref{clm:be}, there is a constant $c>0$ such that $G$ has at least $b \geq cnd^2$ high-degree branching events. Then
    $$
    \sum_{e \in E(J)}w(e) \geq b \geq cnd^2,
    $$
    where $w(e)$ denotes the weight of edge $e \in E(J)$. We will perform the following procedure on $J$. 
    \begin{enumerate}
        \item Delete all edges $e \in E(J)$ with weight $$
        w(e) \leq \frac{cnd^2}{4p^2}.
        $$
        \item While there exists a node $\pi \in V(J)$ such that sum of the edge weights incident to $\pi$ is small, we will delete node $\pi$ from $V(J)$. 
        Formally, for each node $\pi \in V(J)$, let $S_{\pi} \subseteq V(\pi)$ be the set of nodes in path $\pi \in \pi(P)$ with outdegree at least $d/4$, i.e., 
        $$
        S_{\pi} = \{v \in V(\pi) \mid \outdeg_G(v) \geq d/4\}.
        $$
        While there exists a node $\pi \in V(J)$ such that the  sum of the edge weights incident to $\pi$ in $J$ is at most
    $$
    \sum_{\{ (\pi, \pi') \in E(J) \mid \pi' \in V(J) \}} w((\pi, \pi')) \leq \frac{|S_{\pi}|d}{16},
    $$
      we will delete node $\pi$ from $J$.      
    \end{enumerate}
    
    We claim that after the above procedure terminates, graph $J$ is nonempty and still contains a constant fraction of its edge weight. After step 1 of the procedure terminates, $J$ will have total edge weight at least 
    $$
    \sum_{e \in E(J)}w(e)  \geq b - \frac{cnd^2}{4p^2} \cdot p^2 = b - cnd^2/4 = 3/4 \cdot b,
    $$
    since $|E(J)| = p^2$ and $b \geq cnd^2$. 
    Consequently, a large fraction of the edge weights in $J$ survive step 1 of our procedure. Now we will attempt to argue the same is true for the second step of our procedure. 
    
    By the definition of high-degree branching events,
    $$
    b \geq \sum_{\pi \in V(J)} |S_{\pi}| \cdot \frac{d}{4},
    $$
    since there are at each node $v \in S_{\pi}$, there are at least $d/4$ high-degree branching events between $\pi$ and other paths in $V(J)$. Then after step 2 of the procedure terminates, $J$ will have total edge weight at least
    $$
    \sum_{e \in E(J)}w(e)  \geq 3/4 \cdot b - \sum_{\pi \in V(J)} \frac{|S_{\pi}|d}{16} =  3/4 \cdot b - \frac{1}{4} \cdot \left( \sum_{\pi \in V(J)} |S_{\pi}| \cdot \frac{d}{4} \right) \geq b/2.
    $$
    Consequently, after our procedure terminates, the amount of edge weight remaining in $E(J)$ is at least $b/2$, so $E(J)$ and $V(J)$ must be nonempty.

 We have shown that after this procedure terminates, graph $J$ remains nonempty. In particular, we have shown that every surviving node in $V(J)$ is incident to edge weight at least $|S_{\pi}|d/16$, and every surviving edge in $E(J)$ has weight at least $cnd^2/(4p^2)$. Then we can take $\pi^*$ to be any surviving node  $\pi^* \in V(J)$, and we can take $\pi(Q)$ to be the set of all surviving paths adjacent to $\pi^*$ in $J$. The lemma immediately follows.  
\end{proof}

We will denote the path described in Lemma \ref{lem:base_path} as $\pi^*$, and we will use $Q$ to denote the associated set of demand pairs.

\begin{definition}[$\pi$-close nodes]
\label{def:close}
Fix a path $\pi \in \pi(Q)$, and let 
$$
B = \{((s_i, t_i), (s_i, u_i))\}_{i \in [1, k]}
$$
be the set of all high-degree branching events between $\pi$ and $\pi^*$, written so that
\begin{itemize}
    \item $(s_i, t_i) \in \pi$ and $(s_i, u_i) \in \pi^*$, and
    \item If $i < j \in [1, k]$, then node $s_i$ comes before node $s_j$ on path $\pi$, and node $s_j$ comes before node $s_i$ on path $\pi^*$. Note that nodes $s_i$ and $s_j$ cannot appear in the same order on both $\pi$ and $\pi^*$, by consistency.
\end{itemize}
We say that a pair of nodes $(x, y) \in V(\pi^*) \times V(\pi^*)$ is $\pi$-close if $(x, y) = (s_i, s_{i+1})$ for some $i \in [1, k-1]$. 
\end{definition}

\begin{claim}
    If nodes $x$ and $y$ are in $\pi^*$, and $x$ and $y$ are $\pi_1$-close and $\pi_2$-close for paths $\pi_1, \pi_2 \in \pi(Q)$, then $\pi_1 = \pi_2$. 
    \label{clm:distinct_pairs}
\end{claim}
\begin{proof}
    This claim will follow directly from our choice of a \textit{consistent} tiebreaking scheme $\pi(\cdot, \cdot)$ and our definition of branching events. Let $x$ and $y$ be nodes in $\pi^*$ that  are $\pi_1$-close and $\pi_2$-close, for some $\pi_1, \pi_2 \in \pi(P)$. We may assume wlog that node $x$ comes before node $y$ on path $\pi^*$. 
    Since $\pi_1$ has branching events with $\pi^*$ at nodes $x$ and $y$, by consistency we must have that $y$ comes before $x$ on path $\pi_1$. By an identical argument $y$ comes before $x$ on path $\pi_2$. Then by consistency, it follows that $\pi_1[y, x] = \pi_2[y, x]$.
    
    Let $z$ be the node following $y$ in paths $\pi_1$ and $\pi_2$, so that $(y, z) \in \pi$. Let $w$ be the node following $y$ in path $\pi^*$. Observe that $(y, z), (y, w)$ is a branching event. Moreover, by the definition of $\pi$-close nodes, branching event $(y, z), (y, w)$ is associated with paths $\pi_1, \pi^*$ \textit{and} paths $\pi_2, \pi^*$. However, as stated in Definition \ref{def:be}, every branching event is associated with at most one pair of paths. We conclude that $\pi_1 = \pi_2$. 
\end{proof}

\begin{definition}[close nodes]
We say that nodes $x, y \in V$ are \textit{close} if 
$$
    \dist_G(x, y) + \dist_G(y, x) = O\left(\frac{\widehat{\ell}p^2}{nd^2} \right).
$$
\end{definition}

\begin{claim}
    Fix a path $\pi = \pi(s, t) \in \pi(Q)$, and let $C \subseteq V(\pi^*) \times V(\pi^*)$ be the set of all pairs of nodes in $\pi^*$ that are $\pi$-close. Then at least $|C|/2$ node pairs $(x, y) \in C$ are close. 
    \label{clm:pi_close-to-close}
\end{claim}
\begin{proof}
    Fix a path $\pi \in \pi(Q)$, and let 
$$ 
B = \{((s_i, t_i), (u_i, v_i))\}_{i \in [1, k]} 
$$ 
be the set of all high-degree branching events between $\pi$ and $\pi^*$, written so that
\begin{itemize} 
    \item $(s_i, t_i) \in \pi$ and $(u_i, v_i) \in \pi^*$, and 
    \item if $i < j \in [1, k]$, then node $s_i$ comes before node $s_j$ on path $\pi$, and node $s_j$ comes before node $s_i$ on path $\pi^*$. 
\end{itemize} 
Recall from Definition \ref{def:close} that the set of all  pairs of $\pi$-close nodes in $\pi^*$ is exactly $C = \{ (s_i, s_{i+1}) \}_{ i \in [1, k-1]}$, where $|C| = k-1$. Additionally, as we assumed at the beginning of Section \ref{sec:pres_up}, $|\pi| \leq 2 \widehat{\ell}$ and $|\pi^*| \leq 2\widehat{\ell}$. In particular, this implies 
$$
\sum_{i=1}^{k-1} \left( \dist_G(s_i, s_{i+1})    + \dist_G(s_{i+1}, s_i)   \right) \leq |\pi| + |\pi^*| \leq 4\widehat{\ell}.
$$
Let $C_1 \subseteq C$ be the set of all $\pi$-close pairs of nodes $(s_i, s_{i+1})$ with roundtrip distance at least $16\widehat{\ell} / k$, i.e.,
$$
C_1 = \{ (s_i, s_{i+1}) \mid  \dist_G(s_i, s_{i+1}) + \dist_G(s_{i+1}, s_i) \geq 16 \widehat{\ell} /k  \}.
$$
By the previous inequality, we observe that $$|C_1| \cdot 16\widehat{\ell} / k \leq |\pi| + |\pi^*| \leq  4 \widehat{\ell},$$ so $|C_1| \leq k/4$. Let $C_2 = C \setminus C_1$. It follows that $|C_2| \geq |C|/2$. Additionally, by Lemma \ref{lem:base_path}, we know that $|C| = k-1 = \Omega(nd^2/p^2)$. Then for all $(s_i, s_{i+1}) \in C_2$,
$$
\dist_G(s_i, s_{i+1}) + \dist_G(s_{i+1}, s_i) \leq 16\widehat{\ell} / k = O\left(\frac{\widehat{\ell}p^2}{nd^2}\right).
$$
We conclude that all pairs of nodes $(x, y) \in C_2$ are close, establishing the claim. 
\end{proof}

\begin{claim}
\label{clm:many_close}
    Let $S \subseteq V(\pi^*)$ be the set of nodes $v$ in $\pi^*$ of degree at least $\deg_G(v) \geq d/4$. There are $\Omega(|S| d)$ distinct pairs of nodes in $S \times S$ that are close. 
\end{claim}
\begin{proof}
Let $\pi(Q) = \{\pi_1, \dots, \pi_k\}$, for some integer $k$. For $i \in [1, k]$, let $C_i \subseteq V(\pi^*) \times V(\pi^*)$ be the set of all pairs of nodes in $\pi^*$ that are $\pi_i$-close. Note that if $(x, y) \in C_i$, then $\deg_G(x) \geq d/4$ and $\deg_G(y) \geq d/4$ by Definition \ref{def:close}, so $(x, y) \in S \times S$. For $i \in [1, k]$, let $C_i' \subseteq C_i$ be the set of all pairs of nodes $(x, y) \in C_i$ that are close. 
We will need the following statements proved earlier. 
\begin{itemize}
    \item By Claim \ref{clm:distinct_pairs}, $C_i \cap C_j = \emptyset$ for $i \neq j \in [1, k]$. 
    \item By Claim \ref{clm:pi_close-to-close}, $|C_i'| \geq |C_i|/2$ for $i \in [1, k]$. 
    \item By Lemma \ref{lem:base_path}, the total number of high-degree branching events between path $\pi^*$ and the paths in $\pi(Q)$ is at least $\Omega(|S| d)$. This implies that $\sum_{i \in [1, k]}|C_i| = \Omega(|S| d)$.
\end{itemize}
Putting it all together, the total number of distinct close pairs in $S \times S$ is at least
$$
\left|\bigcup_{i \in [1, k]}C_i'\right| = \bigcup_{i \in [1, k]} \left|C_i'\right| \geq \frac{1}{2} \cdot \bigcup_{i \in [1, k]} \left|C_i\right| = \Omega(|S| d).
$$
\end{proof}

We are now ready to prove Lemma \ref{lem:low-diam-scc}.
\begin{proof}[Proof of Lemma \ref{lem:low-diam-scc}]
    Let $S \subseteq V(\pi^*)$ be the set of nodes referenced in Claim \ref{clm:many_close}. By Claim \ref{clm:many_close},  each node $u \in S$ is close to $\Omega(d)$ distinct nodes in $S$ on average. Then there must exist a node $v \in S$ that is close to $\Omega(d)$ distinct nodes in $S$. Let $T$ be the set of nodes in $S$ close to $v$. Observe that for all $x, y \in T$,
    $$
    \dist_G(x, y) \leq \dist_G(x, v) + \dist_G(v, y) = O\left(\frac{\widehat{\ell}p^2}{nd^2}\right).
    $$
    Moreover, by our choice of set $S$ in Claim \ref{clm:many_close}, $\deg_G(x) = \Omega(d)$ for all $x \in S$. We conclude that set $S$ satisfies all properties described in the statement of Lemma \ref{lem:low-diam-scc}. 
\end{proof}

\section{Directed Approximate Preserver Lower Bound}

\label{sec:aprx_pres_lb}

\begin{theorem}
\label{thm:aprx_pres_lb}
For any $n$ and $\alpha = \alpha(n)$, there exist $n$-node directed weighted graphs $G$ with demand pairs $P$ of size $|P| = p$, with $p \in [n^{1/2}, n^2]$, such that any $\alpha$-approximate preserver on $G, P$ has at least
$$
\Omega(n^{2/3}p^{2/3})
$$
edges.
\end{theorem}

\subsection*{Construction}
Let $\mathcal{X}$ and $\mathcal{L}$ be an arrangement of $|\mathcal{X}|=n$ points and $|\mathcal{L}|=p$ distinct lines in $\mathbb{R}^2$ with $\varphi = \Omega(n^{2/3}p^{2/3})$ point-line incidences (e.g., take $\mathcal{X}$ and $\mathcal{L}$ to be the arrangement of points and lines Elekes' construction in \cite{elekes2001sums}). We may assume that every line in $\mathcal{L}$ has finite slope wlog (e.g., by rotating the arrangement of points and lines). 
Given $\mathcal{X}$ and $\mathcal{L}$, we will construct a directed graph $G$ with associated demand pairs $P$. We choose the vertices of $G$ to be the points in $\mathcal{X}$, i.e.,
$$
V(G) = \mathcal{X}.
$$
For each line $L \in \mathcal{L}$, we will add edges to $E(G)$ and demand pairs to $P$ as follows.
\begin{itemize}
\item Let $x_1, \dots, x_k$ denote the points lying on $L$, ordered so that the first coordinates of points are strictly increasing.  
\item Add directed edge $(x_i, x_{i+1})$ to $E(G)$ for $i \in [1, k-1]$.
\item Add demand pair $(x_1, x_k)$ to $P$. 
\end{itemize}
What remains is to assign weights to the edges of $G$.  
Let $L_1, \dots, L_{p}$ denote the lines in $\mathcal{L}$, ordered by increasing slope (breaking ties arbitrarily). If edge $(x, y) \in E(G)$ satisfies $x, y \in L_i$, then we assign edge $(x, y)$ weight
$$
w((x,y)) = \left(2\alpha n \right)^{i}.
$$
Weight $w((x, y))$ is well-defined because distinct lines intersect on at most one point. 

\subsection*{Analysis}

\begin{claim}
Fix a line $L \in \mathcal{L}$. Let $x_1, \dots, x_k$ denote the points lying on $L$, ordered by increasing first coordinate. Then $(x_1, \dots, x_k)$ is the unique $\alpha$-approximate shortest path between $x_1$ and $x_k$ in $G$.
\label{claim:over_under}
\end{claim}
\begin{proof}
Let $L = L_i$ for some $i \in [1, p]$. Let $\pi = (x_1, \dots, x_k)$ denote the $(x_1, x_k)$-path in $G$ traveling along line $L$. Observe that
$$
w(\pi) = \sum_{j=1}^{k-1}w((x_j, x_{j+1})) = k \cdot (2\alpha n)^i \leq  2^i \alpha^i n^{i+1}.
$$
Let $\pi'$ be an $(x_1, x_k)$-path such that $\pi' \neq P$.  
 Let $(x, y) \in E(G)$ be the first edge on path $\pi'$ such that $(x, y) \in E(\pi') \setminus E(\pi)$. If $x, y \in L_j$ for some $j > i$, then
$$
w(\pi') \geq w((x, y)) = (2\alpha)^j n^{j} \geq (2\alpha)^{i+1} n^{i+1} > \alpha \cdot w(\pi),
$$
so $\pi'$ is not an $\alpha$-approximate shortest path.

Otherwise, $x, y \in L_j$ for some $j < i$. Consider the line $L'$ passing through points $y$ and $x_k$. Since $x, y \in L_j$ for some $j < i$, the slope of line $L'$ is strictly greater than the slope of line $L_i$. Additionally, by construction, the first coordinates of the nodes on path $\pi'$ are strictly increasing. 
Then the subpath $\pi'[y, x_k]$ of $\pi'$ must contain an edge  $(x', y') \in E(G)$  such that $x', y' \in L_{j'}$ for some $j' > i$. We conclude that $\pi'$ is not an $\alpha$-approximate shortest path.
\end{proof}

\begin{claim}
Any $\alpha$-approximate preserver on $G, P$ has at least $\varphi - p = \Omega(n^{2/3}p^{2/3})$ edges. 
\label{claim:preserver_lb}
\end{claim}
\begin{proof}
By Claim \ref{claim:over_under}, every edge $(x, y) \in E(G)$ lies on a unique $\alpha$-approximate shortest path between a demand pair $(s, t) \in P$. Then any $\alpha$-approximate preserver of $G$ must contain every edge in $E(G)$, and must be of size at least
$$
|E(G)| = \sum_{L \in \mathcal{L}}  |\{ x \in \mathcal{X} \mid x \in L  \}| - 1  = \varphi - p. 
$$
\end{proof}

\section{Directed Approximate Hopset Lower Bound}
\label{sec:dir_apx_hop_lb}

\begin{theorem}
For any $n$ and $\alpha = \alpha(n)$, there exists an $n$-node directed weighted graph $G$ such that any $O(n)$-size $\alpha$-approximate hopset $H$ of $G$ reduces the hopbound of $G \cup H$ to $\Omega(n^{1/2})$. 
\label{thm:dir_apx_hop_lb}
\end{theorem}

Our directed approximate hopset lower bound construction is based on the undirected exact hopset lower bound of \cite{bodwin2023folklore}, but will use a weighting scheme similar to our approximate preserver lower bound in Section \ref{sec:aprx_pres_lb}. We will assume wlog that $n$ is a power of $4$.

\paragraph{Vertex Set $V$.} Our vertex set $V$ will correspond to integer points in $\mathbb{Z}^2$. Specifically, we take $V$ to be 
$$
V = [ \sqrt{n}/2] \times [2\sqrt{n}].
$$
(Recall that $\sqrt{n}$ is a power of 2, since $n$ is a power of 4.) 
For each $i \in [\sqrt{n} / 2]$, let $L_i = \{(i, y) \in \mathbb{Z}^2 \mid (i, y) \in V\}$ denote the nodes in $V$ with first coordinate $i$. 

\paragraph{Edge Set $E$.} For each node $(x, y) \in V$, we add the directed edges
$$
((x, y), (x+1, y)) \text{ and } ((x, y), (x+1, y+1))
$$
to $E$, if they are contained in $V \times V$. We will now assign edge weights to edges in $E$. 

Our weighting scheme will rely on a permutation  $q$ on $\ints{\sqrt{n} / 2}$ that was used in~\cite{DBLP:conf/soda/WilliamsXX24}. For any integer $i \in \ints{\sqrt{n} / 2}$, $q_i$ is defined as the integer whose $(\log (\sqrt{n} / 2))$-bit binary representation is the reversal of the $(\log (\sqrt{n} / 2))$-bit  binary representation of $i$. For instance $q_0 = 0$ and $q_1 = \sqrt{n} / 4$ (recall $\sqrt{n}$ is a power of $2$).

For each edge $e$ in $L_i \times L_{i+1}$, $i \in [\sqrt{n} - 1]$:
\begin{itemize}
    \item If $e = ((x, y), (x+1, y))$ for some $(x, y) \in L_i$, then assign edge $e$ weight
    $$
w(e) = 1.
$$
    \item Otherwise, $e = ((x, y), (x+1, y+1))$ for some $(x, y) \in L_i$. Assign edge $e$ weight
    $$
    w(e) = (2\alpha n)^{q_{i}}.
    $$
\end{itemize}

\paragraph{Critical Paths $P$.}
We will define a collection of directed paths in $G$ with certain properties that will guarantee that $G$ implies an $\alpha$-approximate hopset lower bound. 

For each node $v \in [1] \times [\sqrt{n}]$ and each $d \in [\sqrt{n}]$, we iteratively define a critical path $\pi$ as follows. The starting node of path $\pi$ is $v$. If path $\pi$ contains its $i$th node $v_i$ in layer $L_i$, then we define its $(i+1)$th node $v_{i+1}$ in layer $L_{i+1}$ as follows: 
\begin{itemize}
    \item If $d  \le q_i$, then let 
    $$
    v_{i+1} = v_i + (1, 0),
    $$
    if $v_{i+1} \in V$ (in fact, this condition always holds).
    \item If $d > q_i$, then let
    $$
    v_{i+1} = v_i + (1, 1),
    $$
    if $v_{i+1} \in V$ (in fact, this condition always holds). 
\end{itemize}
The path $\pi$ begins at node $v_1 = v \in L_1$, and terminates at node $v_{\sqrt{n}/2} \in L_{\sqrt{n}/2}$. We add $\pi$ to our collection of paths $P$. 

\begin{lemma}
The above graph has the following properties:
\begin{enumerate}
    \item \label{item:lem:prop_dir_apx_hop:item1} $|V| = n$, $|E| = \Theta(n)$, and $|P| = \Theta(n)$.
    \item \label{item:lem:prop_dir_apx_hop:item2} Every path $\pi \in P$  contains exactly one node in each layer $L_i$, $i \in [\sqrt{n}/2]$. 
    \item \label{item:lem:prop_dir_apx_hop:item3} Every path $\pi \in P$ is the unique $\alpha$-approximate shortest path between its endpoints. 
    \item \label{item:lem:prop_dir_apx_hop:item4} For any $k$ distinct paths $\pi_1, \dots, \pi_k \in P$, where $k \in [|P|]$, the intersection $\cap_{i=1}^k \pi_i$ has size at most 
    $$
    \left| \bigcap_{i=1}^k \pi_i \right| \leq \frac{2\sqrt{n}}{k}.
    $$
\end{enumerate}
\label{lem:prop_dir_apx_hop}
\end{lemma}
\begin{proof} We prove each of these items below.
\begin{enumerate}
    \item The size bounds on $|V|$ and $|E|$ follow directly from the construction. In our construction of critical paths $P$, we add $\Theta(n)$ paths to $P$, but it is not immediate that these paths are distinct. Note that if two paths have distinct parameters $d_1 < d_2$ respectively, they will use different edges between $L_{j}$ and $L_{j+1}$, for $q_j = d_1$. If $d_1 = d_2$, then the two paths begin at different vertices on $L_1$. Thus, $|P| = \Theta(n)$.

    \item Each path $\pi \in P$ begins at a vertex $v_1 = (1, y) \in [1] \times [\sqrt{n}]$. Each edge in $\pi$ corresponds to a vector of form $(1, 0)$ or $(1, 1)$ in $G$. Then the $i$th node $v_i$ of $\pi$ (if it exists) will satisfy
    $$
    v_i \in  \{i \} \times [y, y + i] \subseteq \{i\} \times [\sqrt{n}, \sqrt{n}+i] \subseteq V,
    $$
    for $i \in [\sqrt{n}/2]$. We conclude that $\pi$ contains a node in each layer $L_i$, $i \in [1, \sqrt{n}/2]$. Moreover, $\pi$ can contain at most one node in each layer $L_i$ because every edge in $G$ is a directed edge in $L_i \times L_{i+1}$. 

    \item Fix a path $\pi \in P$ with start node $s \in L_1$ and end node $t \in L_{\sqrt{n}/2}$. Let $d \in [\sqrt{n}]$ be the parameter used to construct $\pi$. 
    Consider an $s \leadsto t$ path $\pi'$, where $\pi' \neq \pi$. Suppose that $\pi'$ contains an edge $e' \in L_i \times L_{i+1}$ corresponding to vector $(1, 1)$, and $d \le q_i$. Then
    $$
    w(\pi') \geq w(e') = (2\alpha n)^{q_i} \geq (2\alpha n)^{d} \geq 2\alpha \cdot (n \cdot (2\alpha n)^{d-1}) \geq 2\alpha \cdot w(\pi),
    $$
    where the final inequality follows from the fact that every edge in $\pi$ has weight at most $(2\alpha n)^{d-1}$. Then path $\pi'$ is not an $\alpha$-approximate shortest $s \leadsto t$ path. 

    This implies that any $\alpha$-approximate shortest $s \leadsto t$ path $\pi'$ takes an edge $e' \in L_i \times L_{i+1}$ corresponding to vector $(1, 1)$, only if path $\pi$ takes an edge $e \in L_i \times L_{i+1}$ corresponding to vector $(1, 1)$. Moreover, since $\pi$ and $\pi'$ are both $s \leadsto t$ paths, they must take exactly the same number of edges corresponding to vector $(1, 1)$. We conclude that path $\pi'$ takes an edge $e' \in L_i \times L_{i+1}$ corresponding to vector $(1, 1)$ if and only if path $\pi$ takes an edge $e \in L_i \times L_{i+1}$ corresponding to vector $(1, 1)$. Then $\pi = \pi'$, completing the proof.
    
    \item Let $d_1, \ldots, d_k$ be the parameters used to construct $\pi_1, \ldots, \pi_k$. First, observe that two distinct paths sharing the same parameter $d$ cannot intersect, so all values $d_1, \ldots, d_k$ are distinct. wlog, we assume $d_1 \le d_k - k + 1$. Second, by \cref{item:lem:prop_dir_apx_hop:item3}, all paths $\pi_1, \ldots, \pi_k$ are unique shortest paths, so $\cap_{i=1}^k \pi_i$ must be a contiguous path. 

    Suppose for some $i\in [\sqrt{n}/2 - 1]$, $d_1 \le q_i < d_k$, then $\pi_1$ and $\pi_k$ cannot use the same edge between $L_i$ and $L_{i+1}$, implying $\cap_{i=1}^k \pi_i$ also does not contain such an edge. The interval $[d_1, d_k)$ must contain an interval of the form $\mathcal{I} = [a \cdot 2^b, (a+1) \cdot 2^b)$ (a dyadic interval) where $2^b \ge k / 4$. By construction of the permutation $q$, for all values of $i$ whose last $(\log((\sqrt{n} / 2))-b)$-bit binary representation is the reversal of $a$, $q_i$ falls in the interval $\mathcal{I}$, and hence $d_1 \le q_i < d_k$. Regardless where $\cap_{i=1}^k \pi_i$ begins, the next such $i$ is within $2^{\log((\sqrt{n} / 2))-b} \le \frac{2\sqrt{n}}{k}$ layers away, at which point $\cap_{i=1}^k \pi_i$ must end. Therefore, $|\cap_{i=1}^k \pi_i| \le \frac{2\sqrt{n}}{k}$.
\end{enumerate}
\end{proof}

Using the hopset lower bound argument of \cite{bodwin2023folklore}, we can use Lemma \ref{lem:prop_dir_apx_hop} to argue that any $0.01n$-size $\alpha$-approximate hopset $H$ of $G$ has hopbound $\Omega(\sqrt{n})$. In order to extend this lower bound to any hopset $H$ of size $|H| \leq cn$ for an arbitrarily large constant $c > 0$, we will need to modify our construction of $G$ and $P$, as specified in the following lemma.

\begin{lemma}
For any $n, c > 0$, there exists a directed weighted graph $G' = (V', E')$ and associated collection of paths $P'$ satisfying the following properties:
\begin{enumerate}
    \item \label{item:lem:dir_apx_hop_lb_c:item1} $|V'| = \Theta(n/c)$, and $|P'| = \Theta(n)$. 
    \item \label{item:lem:dir_apx_hop_lb_c:item2}  Every path $\pi \in P'$ is of length $\Theta(n/c)$. 
    \item \label{item:lem:dir_apx_hop_lb_c:item3}  Every path $\pi \in P'$ is the unique $\alpha$-approximate shortest path between its endpoints.
    \item \label{item:lem:dir_apx_hop_lb_c:item4}  For any $k$ distinct paths $\pi_1, \dots, \pi_k \in P$, where $k \in [|P|]$, the intersection $\cap_{i=1}^k \pi_i$ has size at most 
    $$
    \left| \bigcap_{i=1}^k \pi_i \right| \leq \frac{2 \sqrt{n}}{ck}.
    $$
\end{enumerate}
\label{lem:dir_apx_hop_lb_c}
\end{lemma}
\begin{proof}
We will first construct $G', P'$ and then analyze them.
    \paragraph{Construction.} We will choose $V'$ to be a subset of $V$. Specifically, we define $V'$ to contain every $ci$-th layer $L_{ci}$ of $G$ for $i \in [\lfloor \sqrt{n}/(2c) \rfloor]$, i.e.,  
    $$
V' = L_c \cup L_{2c} \cup \dots  \cup L_{\lfloor \sqrt{n}/(2c) \rfloor \cdot c }.
$$
For each path $\pi \in P$ and each index $i \in [\lfloor \sqrt{n}/(2c) \rfloor - 1]$, we will add an edge to $E'$ as follows. 
\begin{itemize}
    \item Note that $\pi$ contains exactly one node in layer $L_{ci}$ and exactly one node in layer $L_{c(i+1)}$ by Property 2 of Lemma \ref{lem:prop_dir_apx_hop}. Let $u = \pi \cap L_{ci}$, let $v = \pi \cap L_{c(i+1)}$. 
    \item Add edge $(u, v)$ to $E'$, and assign $(u, v)$ weight equal to the distance in $G$ from $u$ to $v$, i.e., 
    $$
    w((u, v)) = \dist_{G}(u, v).
    $$
\end{itemize}
For each path $\pi \in P$, we define a path $\pi'$ in $G'$ by removing all nodes in $\pi \setminus V'$ from $\pi$ and taking $\pi'$ to be the resulting path. Path $\pi'$ is well-defined in $G'$ by our construction of $E'$. 
    \paragraph{Analysis.} \cref{item:lem:dir_apx_hop_lb_c:item1,item:lem:dir_apx_hop_lb_c:item2,item:lem:dir_apx_hop_lb_c:item3} follow directly from Lemma \ref{lem:prop_dir_apx_hop} and the observation that for all $u, v \in V'$, $\dist_{G'}(u, v) = \dist_G(u, v)$. 
    To prove \cref{item:lem:dir_apx_hop_lb_c:item4}, note that if $k$ paths $\pi_1', \dots, \pi_k' \in P'$ have an intersection of size $|\cap_i \pi_i'| > \frac{2 \sqrt{n}}{ck}$, then there exist $k$ paths $\pi_1, \dots, \pi_k \in P$ with an intersection of size $|\cap_i \pi_i| > \frac{2 \sqrt{n}}{k}$, contradicting \cref{item:lem:prop_dir_apx_hop:item4} of Lemma \ref{lem:prop_dir_apx_hop}. \qedhere
\end{proof}

We are now ready to finish the proof of Theorem \ref{thm:dir_apx_hop_lb}. 

\begin{proof}[Proof of Theorem \ref{thm:dir_apx_hop_lb}]
Let graph $G' = (V', E')$ and paths $P'$ be constructed with parameters $n, c > 0$. We will show that for any hopset $H$ of size $|H| \leq n = O(c|V'|)$, $G\cup H$ has hopbound $\Omega(\sqrt{n}/c)$, completing the proof.  We use $\hopdist_{G'}(u, v)$ to denote the $\alpha$-approximate hopdistance from $u$ to $v$ in $G'$, and for each $\pi \in P'$, we let $s_{\pi}$ and $t_{\pi}$ denote its start and end node, respectively.

Similar to the proof of \cite{bodwin2023folklore}, we define a potential function with respect to any hopset $H$ (with a slight abuse of notation, we use $G' \cup H$ to denote $(V', E' \cup H)$ in the following):
$$
\Phi(H) = \sum_{\pi \in P'} \hopdist_{G' \cup H}(s_{\pi}, t_{\pi}).
$$
Initially, $\Phi(\emptyset) = |P'| \cdot \Theta(\sqrt{n}/c) = \Theta(n^{3/2}/c)$, since every path $\pi \in P'$ is a unique $\alpha$-approximate shortest path between its endpoints in $G'$ and has hopbound $|\pi| = \Theta(\sqrt{n}/c)$.  We will now upper bound the quantity $\Phi(H) - \Phi(H \cup \{(u, v)\})$ for all $(u, v) \in V' \times V'$ in the transitive closure of $G'$. This quantity corresponds to the decrease in potential after adding one additional edge to our hopset. 

Observe that the hopset edge $(u, v)$ only decreases the hoplength of a path $\pi \in P$ if $u, v \in \pi$. Let $\pi_1, \dots, \pi_k \in P'$ be the paths in $P'$ containing nodes $u$ and $v$. Then $\hopdist_{G'}(u, v) \leq \frac{2\sqrt{n}}{ck}$, by \cref{item:lem:dir_apx_hop_lb_c:item4} of Lemma \ref{lem:dir_apx_hop_lb_c}. It follows that
$$\Phi(H) - \Phi(H \cup \{(u, v)\}) \leq k \cdot \frac{2 \sqrt{n}}{ck} = O(\sqrt{n}/c).$$
Then for any hopset $H$ of size $|H| \leq n$, it follows that  $$
\Phi(H) \geq \Phi(\emptyset) - n \cdot O(\sqrt{n}/c) = \Omega(n^{3/2}/c).
$$
Then by an averaging argument, there is a path $\pi \in P'$ such that the $\alpha$-approximate hopdistance between $s_{\pi}$ and $t_{\pi}$ in $G' \cup H$ is $\Phi(H) / |P'| = \Omega(\sqrt{n}/c)$. We conclude that the $\alpha$-approximate hopbound of $G' \cup H$ is $\Omega(\sqrt{n}/c)$ as desired. 
\end{proof}

\section{\texorpdfstring{$O(m)$}{O(m)}-size Shortcut Set Lower Bound}
\label{sec:shortcut}

\begin{theorem}
\label{thm:shortcut}
There exist $n$-node directed graphs $G$ with $|E(G)| = m$ edges, such that any $O(m)$-size shortcut set $H$ of $G$ reduces the diameter of $G \cup H$ to $\Omega(n^{2/9})$.
\end{theorem}

Our construction is parameterized by two parameters $r$ and $c$, where $r$ is a parameter controlling the size of the graph already used in \cite{DBLP:conf/soda/WilliamsXX24} and $c$ intuitively denotes the number of copies of their construction we use. We assume $r$ is a power of $2$ wlog.

\paragraph{Vertex Set $V$. } Our graph has $2cr + 1$ layers, indexed by integers between $0$ and $2cr$. Each layer contains vertices from $[4 cr] \times [4 cr] \times [4 cr^2]$, which can be viewed as points in the $3D$ grid. For $i \in \ints{2cr + 1}$ and $\mathbf{x} \in [4 cr] \times [4 cr] \times [4 cr^2]$, we also index the vertex with coordinates $\mathbf{x}$ in the $i$-th layer by $(i, \mathbf{x})$. 

\paragraph{Edge Set $E$. } 
Similar to \cref{sec:dir_apx_hop_lb}, we also use a permutation $q$ on $\ints{r}$ that was used in~\cite{DBLP:conf/soda/WilliamsXX24}. For any integer $i \in \ints{r}$, $q_i$ is defined as the integer whose $(\log r)$-bit binary representation is the reversal of the $(\log r)$-bit  binary representation of $i$. For instance $q_0 = 0$ and $q_1 = r / 2$ (recall $r$ is a power of $2$). For every even $i \in \ints{2cr}$, $\mathbf{w}_i = (1, 0, q_{(i / 2) \bmod{r}})$; for every odd $i \in \ints{2cr}$, $\mathbf{w}_i = (0, 1, q_{((i - 1) / 2) \bmod{r}})$. Note that the sequence of vectors $\{\mathbf{w}_i\}_i$ has period $2r$. 

Then for every $i \in \ints{2cr}$, and every point $\mathbf{x} \in [4 cr] \times [4 cr] \times [4 cr^2]$, we add an edge from $(i, \mathbf{x})$ to $(i + 1, \mathbf{x})$, and an edge from $(i, \mathbf{x})$ to $(i + 1, \mathbf{x} + \mathbf{w}_i)$ if $\mathbf{x} + \mathbf{w}_i$ stays inside the grid. 

\paragraph{Critical Paths $P$. } We define a set of paths in the graph, which has certain desired properties. Eventually, we will argue that at least one of these paths will have long distances between its endpoints even if we add $O(m)$ shortcut edges to the graph. 

For every $\mathbf{x} \in [2 cr] \times [2 cr] \times [2 cr^2]$ and every $(d_1, d_2) \in [r] \times [r]$, we define a critical path as follows. The start node of the path is $(0, \mathbf{x})$, and the path travels from layer $0$ to layer $2cr$. Between layer $i$ and layer $i+1$, there are two potential edges for this path to travel, one corresponding to $\mathbf{w}_i$, and one corresponding to $\mathbf{0} = (0, 0, 0)$. The path takes the edge corresponding to $\mathbf{w}_i$ if $\mathbf{w}_i \in (\{1\} \times \{0\} \times \ints{d_1}) \cup (\{0\} \times \{1\} \times \ints{d_2})$ (i.e., if $\mathbf{w}_i$ is in the form of $(1, 0, z)$, then we require $z < d_1$; otherwise, $\mathbf{w}_i$ is in the form of $(0, 1, z)$, and we require $z < d_2$), otherwise, the path takes the edge corresponding to $\mathbf{0}$. It is not difficult to verify that the critical paths defined this way do not fall out of the grid. 

\begin{lemma}
\label{lem:shortcut-properties}
The above graph has the following properties:
\begin{enumerate}
    \item \label{item:lem:shortcut-properties:item1} The number of vertices $|V| = \Theta(c^4 r^5)$, the number of edges $|E| = \Theta(c^4 r^5)$, the number of critical paths $|P| = \Theta(c^3 r^6)$. 
    \item \label{item:lem:shortcut-properties:item2}
    Every path $\pi \in P$ is the unique path from its start node to its end node. 
    \item \label{item:lem:shortcut-properties:item3}
    For any path $\sigma$ in the graph of length $g \le 2r$, there are $O((r / g)^2 )$ critical paths in $P$ containing $\sigma$. 
\end{enumerate}
\end{lemma}
\begin{proof}
We prove each of the items below. 
\begin{enumerate}
    \item These  bounds are elementary. 
    \item Fix an arbitrary critical path $\pi$ from $(0, \mathbf{x})$ to $(2cr, \mathbf{y})$, associated with $(d_1, d_2) \in [r] \times [r]$. Each path from $(0, \mathbf{x})$ to $(2cr, \mathbf{y})$ corresponds to a unique subset $S$ of $\ints{2cr}$, i.e., if $i \in S$, the path uses the edge corresponding to $\mathbf{w}_i$ between layer $i$ and layer $i+1$ and otherwise it uses the edge corresponding to $\mathbf{0}$. Thus, it suffices to argue that there is a unique subset $S$ of $\ints{2cr}$ so that $\sum_{i \in S} \mathbf{w}_i = \mathbf{y} - \mathbf{x}$. 

    We have $c$ copies of the vector $(1, 0, z)$ for $z \in \ints{r}$, and $c$ copies of the vector $(0, 1, z)$ for $z \in \ints{r}$. By construction, $\mathbf{y} - \mathbf{x} = (c d_1, c d_2, c \cdot \frac{d_1 (d_1 - 1)}{2} + c \cdot \frac{d_2 (d_2 - 1)}{2})$. In order for any subset of $\{\mathbf{w}_i\}_i$ to be equal to $\mathbf{y} - \mathbf{x}$, we must pick exactly $c d_1$ vectors of the form $(1, 0, z)$, and $c d_2$ vectors of the form $(0, 1, z)$. Also, $c \cdot \frac{d_1 (d_1 - 1)}{2} + c \cdot \frac{d_2 (d_2 - 1)}{2}$ is the smallest sum of obtainable by picking such a subset of vectors, and it is simple to see that no other subset can achieve the same sum. 

    \item First of all, it is simple to see that two critical paths with the same associated values of $(d_1, d_2)$ do not share any vertices. Therefore, it suffices to bound the number of distinct $(d_1, d_2)$ so that there is a critical path associated with $(d_1, d_2)$ that uses $\sigma$. We additionally assume $g \ge 2$, as otherwise, the claimed bound is simply $O(r^2)$, which is the total number of distinct $(d_1, d_2)$. 

    Now consider a path $\pi$ associated with value $(d_1, d_2)$ and another path $\pi'$ associated with value $(d_1', d_2')$ where $|d_1 - d_1'| > 16 r / g$. We will show that $\pi$ and $\pi'$ cannot both contain $\sigma$ as a subpath. First, because $\sigma$ has length $g$, it must travel at least $(g-1)/2$ layers where the corresponding vector $w_i$ has the form $(1, 0, z)$ for $z \in \ints{r}$, where $z$ can take value $q_{j \bmod {r}}$ for $j$ in an interval $\mathcal{I}$ of length at least $(g-1)/2$. This interval must contain a dyadic interval $\mathcal{I}'$ of length at least $(g-1)/8$ (dyadic intervals are intervals of the form $[a \cdot 2^k, (a+1) \cdot 2^k)$ for nonnegative integers $a, k$). Say $\mathcal{I}' = [a \cdot 2^k, (a+1) \cdot 2^k)$ for some nonnegative integers $a, k$. Then $\{q_{i \bmod{r}}\}_{i \in \mathcal{I}'}$ contains all numbers from $\ints{r}$ that is congruent to $\hat{a}$ modulo $\frac{r}{2^k}$, where $\hat{a}$ is the $\log(\frac{r}{2^k})$-bit reversal of $a$. Now, suppose for the sake of contradiction that $|d_1 - d'_1| > 16r / g$. Then because $2^k = |\mathcal{I}'| \ge (g-1) / 8 \ge g/16$, we must have $|d_1 - d'_1| > \frac{r}{2^k}$. Then there must exist some $z = q_i$ for some $i \in \mathcal{I}'$ such that $z \in [\min\{d_1, d_1'\}, \max\{d_1, d_1'\})$. By construction, exactly one of $\pi$ and $\pi'$ uses the edge corresponding to $\mathbf{w}_i = (1, 0, z)$ and the other uses the edge corresponding to $\mathbf{0}$, so they cannot both contain $\sigma$ as a subpath. This is a contradiction, so we must have $|d_1 - d'_1| \le 16r / g$. 

    Similarly, we can show that $|d_2 - d'_2| \le 16r / g$. Thus, if a path $\pi'$ with value $(d_1', d_2')$ needs to contain $\sigma$ together with the path $\pi$ with value $(d_1, d_2)$, we must have $|d_1 - d'_1| \le 16r / g$ and $|d_2 - d'_2| \le 16r / g$, so the number of such distinct values $(d_1', d_2')$ is $O((r/g)^2)$. By previous discussions, this imply that the number of critical paths containing $\sigma$ is $O((r/g)^2)$. 
\end{enumerate}
\end{proof}

\begin{proof}[Proof of \cref{thm:shortcut}]
    We consider the graph $G=(V, E)$ with the critical paths $P$ constructed above. 

    We let $E_\Delta$ be the set of shortcut edges in $G$ that connects every pair of vertices $(u, v)$ where $u$ can reach $v$ and $u$ and $v$ are at most $\Delta$ layers away for some parameter $\Delta$.
    For any critical path $\pi \in P$, we use $s_\pi$ to denote its start node and $t_\pi$ to denote its end node. 

    We define a potential function with respect to any shortcut set $H$:
    \[
    \Phi(H) := \sum_{\pi \in P} \dist_{G \cup E_\Delta \cup H}(s_\pi, t_\pi). 
    \]
    Initially, we clearly have $\Phi(\emptyset) = |P| \cdot cr / \Delta = \Theta(c^4 r^7 / \Delta)$. 

    Now we upper bound $\Phi(H) - \Phi(H \cup \{e\})$, i.e., the potential drop after we add one additional edge to the shortcut set. There are several cases to consider:
    \begin{itemize}
        \item $e$ connects two vertices that are at most $\Delta$ layers apart. In this case, $E \cup E_\Delta \cup H = E \cup E_\Delta \cup H \cup \{e\}$, as $e \in E_\Delta$, so $\Phi(H) - \Phi(H \cup \{e\}) = 0$. 
        \item $e$ connects two vertices $u, v$ that are $g > \Delta$ layers apart. Consider any path $\pi \in P$ that contains both $u$ and $v$. It is not difficult to see that $\dist_{G \cup E_\Delta \cup H}(s_\pi, t_\pi) - \dist_{G \cup E_\Delta \cup H \cup \{e\}}(s_\pi, t_\pi) \le g/\Delta$. As there is a unique path in $G$ from $s_\pi$ to $t_\pi$ by \cref{lem:shortcut-properties} \cref{item:lem:shortcut-properties:item2}, there must also be a unique path $\sigma$ in $G$ from $u$ to $v$. Thus, by \cref{lem:shortcut-properties} \cref{item:lem:shortcut-properties:item3}, there are $O(1 + (r/g)^2)$ critical paths in $P$ containing $\sigma$, which are exactly the critical paths that contain both $u$ and $v$. Therefore, the total potential drop in this case is 
        \[
        O(1 + (r/g)^2) \cdot \frac{g}{\Delta} = O\left(\frac{g}{\Delta} + \frac{r^2}{g \Delta}\right),
        \]
        which is always upper bounded by
        \[
        O\left(\frac{cr}{\Delta} + \frac{r^2}{\Delta^2}\right)
        \]
        for $g \in (\Delta, 2cr]$. 
    \end{itemize}
    Therefore, if $|H| = \eps_1 \min\left\{ c^3 r^6, c^4 r^5 \Delta \right\}$ for some sufficiently small constant $\eps_1 > 0$, total potential drop $\Phi(\emptyset) - \Phi(H)$ would be bounded by $\eps_2 c^4 r^7 / \Delta$ for some small constant $\eps_2$. Furthermore, the parameter $\eps_2$ can be made arbitrarily small depending on $\eps_1$, and when it is sufficiently small, we will have $\Phi(H) = \Theta(c^4 r^7 / \Delta)$. In this case, by averaging, there must exist a path $\pi \in P$ where $\dist_{G \cup H}(s_\pi, t_\pi) \ge \dist_{G \cup E_\Delta \cup H}(s_\pi, t_\pi) = \Theta(cr / \Delta)$, i.e., the diameter of the graph is at least $\Theta(cr / \Delta)$ after adding a shortcut set of size $|H|=\eps_1 \min\left\{ c^3 r^6, c^4 r^5 \Delta \right\}$. 

    By setting $c = r / B$ for some large constant $B$, we get that $|V|, |E| = \Theta(r^9 / B^4)$ by \cref{lem:shortcut-properties} \cref{item:lem:shortcut-properties:item1}, and the number of edges allowed in the shortcut set is $\eps_1 \min\left\{r^9 / B^3, r^9 \Delta  / B^4 \right\}$, which can be any constant factor larger than $|E|$ by setting $B$ and $\Delta$ big enough. The diameter bound will be $\Theta(r^2 / B\Delta)$, which is $\Theta(|V|^{2/9})$ for constant values of $B$ and $\Delta$. 

\end{proof}

\section{\texorpdfstring{$O(m)$}{O(m)}-size Exact Hopset Lower Bound in Unweighted Graphs}
\label{sec:hopset}

The following construction is similar to the construction in \cref{sec:shortcut}, and the main difference is that the construction in \cref{sec:shortcut} is a layered graph, while the construction in this section is an unlayered graph. In comparison, we are able to utilize an unlayered graph construction in this section because it suffices to ensure unique shortest paths have good structures in a construction for hopset lower bound, while the construction in \cref{sec:shortcut} requires unique paths to have good structures. The unlayered construction complicates the analysis, but leads to an improved lower bound. 

\begin{theorem}
\label{thm:hopset}
There exist $n$-node undirected unweighted graphs $G$ with $|E(G)| = m$ edges, such that any $O(m)$-size exact hopset $H$ of $G$ reduces the hopbound of $G \cup H$ to $\Omega(n^{2/7})$.
\end{theorem}

Similar as before, there are two parameters $r$ and $c$, and we assume $r$ is a power of $2$. 

\paragraph{Vertex Set $V$. } The important vertices in the graph are indexed by $\ints{2cr + 1} \times [4cr] \times [3cr^2]$. We will say that a vertex $(i, j, k)$ is on the $i$-th layer. There might also be other (unimportant) vertices defined in the following.

\paragraph{Edge Set $E$. } Recall $q$ is a permutation on $\ints{r}$. 
For $i \in \ints{2cr}, j \in [4cr], k \in [3cr^2]$, we add the following edges (if some edges fall out of the grid, then we do not add them):
\begin{itemize}
    \item An edge between $(i, j, k)$ and $(i + 1, j, k)$;
    \item A chain of two edges between $(i, j, k)$ and $(i + 1, j, k + q_{(i/2) \bmod{r}})$  if $i$ is even (this creates a new vertex). 
    \item A chain of two edges between $(i, j, k)$ and $(i + 1, j + 1, k + q_{((i-1)/2)\bmod{r}})$  if $i$ is odd (this creates a new vertex). 
\end{itemize}

\paragraph{Critical Paths $P$. } Each critical path is associated with a start node $(0, j, k)$ for some $j \in [c r]$ and $k \in [cr^2]$, together with a pair of values $(d_1, d_2) \in [r] \times [r]$ where $r / 2 < d_2 < d_1$. When at a vertex $(i, j', k')$ for some $j', k'$ and even $i$, the path will travel to $(i + 1, j', k' + q_{(i / 2) \bmod{r}})$ (using a chain of two edges) if $q_{(i / 2) \bmod{r}} \ge d_1$; otherwise, it will travel to $(i + 1,j', k')$ (using one edge). When at a vertex $(i, j', k')$ for some $j', k'$ and odd $i$, the path will travel to $(i + 1, j' + 1, k' + q_{((i - 1) / 2) \bmod{r}})$ (using a chain of two edges) if $q_{((i - 1) / 2) \bmod{r}} \ge d_2$; otherwise, it will travel to $(i + 1, j', k')$ (using one edge). The path stops when it reaches some vertex on the $(2cr)$-th layer. It is not difficult to check that the above rules uniquely determine the path.

\begin{lemma}
\label{lem:hopset-properties}
The above graph has the following properties:
\begin{enumerate}
    \item \label{item:lem:hopset-properties:item1} The number of vertices $|V| = \Theta(c^3 r^4)$, the number of edges $|E| = \Theta(c^3 r^4)$, the number of critical paths $|P| = \Theta(c^2 r^5)$. 
    \item \label{item:lem:hopset-properties:item2}
    Every path $\pi \in P$ is the unique shortest path from its start node to its end node. 
    \item \label{item:lem:hopset-properties:item3}
    For any shortest path $\sigma$ in the graph of length $g \le 2r$, there are $O((r / g)^2 )$ critical paths in $P$ containing $\sigma$. 
\end{enumerate}
\end{lemma}
\begin{proof}
    We consider the three items separately:
    \begin{itemize}
        \item These bounds follow straightforwardly from the construction. 
        \item Fix any $\pi \in P$ with start node $s_\pi = (0, j, k) $, end node $t_\pi$ and associated values $(d_1, d_2)$. We will argue that there is no other path from $s_\pi$ to $t_\pi$ using fewer or same number of edges. 

        First, if a path only uses one of the two edges in the chains of two edges we added, and then travels back, then it can never be a shortest path. Thus, for the purpose of this proof, we can view the chains of two edges as single edges with cost $2$. 

        For any path travelling from $s_\pi$ to $t_\pi$ and for any fixed $i \in \ints{2cr}$, the number of times it crosses from layer $i$ (the set of vertices whose first coordinate equal $i$) to layer $i + 1$ must be exactly one more than the number of times it crosses from layer $i+1$ back to layer $i$. As a result, we can pair up each edge going from layer $i$ to layer $i+1$ with an edge going from layer $i + 1$ to layer $i$, so that after the pairing up, we are left with one edge going from layer $i$ to layer $i+1$. Each of these set of edges (either a pair of two edges or a single edge left at the end) changes the coordinate of the current vertex, i.e., add a vector to the coordinate of the vertex, and we consider the average (over the total costs of edges in the set) dot product between the vector and $\delta = (1, 1 - \frac{d_2 - 0.5}{d_1 - 0.5}, \frac{1}{d_1 - 0.5})$. 

        \begin{enumerate}
            \item A single forward edge:
            \begin{enumerate}
                \item $(1, 0, 0)$: the dot product between it and $\delta$ is $1$. 
                \item $(1, 0, x)$: the dot product between it and $\delta$ is $1 + \frac{x}{d_1 - 0.5}$, and the average per cost is $(1 + \frac{x}{d_1 - 0.5})/2$, which is greater than $1$ when $x \ge d_1$, and is less than $1$ when $x < d_1$. 
                \item $(1, 1, y)$: the dot product between it and $\delta$ is $2 -\frac{d_2-0.5}{d_1 - 0.5} + \frac{y}{d_1 - 0.5}$, and the average per cost is $(2 + \frac{y-d_2 + 0.5}{d_1 - 0.5})/2$, which is greater than $1$ when $y \ge d_2$, and is less than $1$ when $y < d_2$. 
            \end{enumerate}
            \item One arbitrary forward edge and one arbitrary backward edge: 
            \begin{enumerate}
                \item these two edges are the negations of each other: in this case, the sum of the two vectors corresponding to the two edges is $0$, and thus the average dot product between it and $\delta$ per cost is $0$, which is less than $1$. 
                \item $(1, 0, 0) - (1, 0, x)$: in this case, the total cost of the two edges is $3$, so the average dot product between it and $\delta$ per cost is $(-\frac{x}{d_1 - 0.5}) / 3 < 0 < 1$. 
                \item $(1, 0, x) - (1, 0, 0)$: in this case, the total cost of the two edges is $3$, so the average dot product between it and $\delta$ per cost is $(\frac{x}{d_1 - 0.5}) / 3$. Recall $d_1 > r / 2$, so $d_1 - 0.5 \ge r / 2$, which implies $(\frac{x}{d_1 - 0.5}) / 3 \le (\frac{r}{r/2})/3 < 1$.  
                \item $(1, 0, 0) - (1, 1, y)$: in this case, the total cost of the two edges is $3$, so the average dot product between it and $\delta$ per cost is $(-1 + \frac{d_2 - 0.5}{d_1 - 0.5}-\frac{y}{d_1 - 0.5}) / 3 < 0 < 1$ (recall $d_2 < d_1$). 
                \item $(1, 1, y) - (1, 0, 0)$: in this case, the total cost of the two edges is $3$, so the average dot product between it and $\delta$ per cost is $(\frac{y+ d_2 - 0.5}{d_1 - 0.5}) / 3$. Recall $d_1 > r/ 2$ so $d_1 - 0.5 \ge r/2$ and $d_2 < d_1$, so $(\frac{y+ d_2 - 0.5}{d_1 - 0.5}) / 3 < (\frac{y}{d_1 - 0.5} + 1) / 3 \le (\frac{r}{r / 2} + 1) / 3 \le 1$. 
            \end{enumerate}
        \end{enumerate}
    Summarizing the above, the only cases where the average dot product is greater than $1$ is $(1, 0, x)$ for $x \ge d_1$ and $(1, 1, y)$ for $y \ge d_2$ (corresponding to a chain of two edges). The only cases where the average dot product is equal to $1$ is $(1, 0, 0)$ (corresponding to a single edge). All other cases have average dot product less than $1$. The critical path we defined from $s_\pi$ to $t_\pi$ exactly use the edges (or set of edges) whose average dot product per cost is the largest. By construction, if there is an alternative path from $s_\pi$ to $t_\pi$, then it must use a different collection of edge sets. Furthermore, the total dot product between $\delta$ and all the vectors corresponding to the edges on the alternative path must be equal to the total dot product between $\delta$ and all the vectors corresponding to the edges on the critical path. As the critical path uses edges that maximize the average dot product per cost, it implies that the alternative path must have a larger cost. Thus, we conclude that the critical path is the unique shortest path between $s_\pi$ and $t_\pi$. 
    \item First, if two critical paths share the same values $(d_1, d_2)$ and start from distinct vertices on the first layer, they will be disjoint. Thus, it suffices to bound the number of distinct $(d_1, d_2)$, so that there could potentially be a critical path with value $(d_1, d_2)$ that contains $\sigma$. 

    Furthermore, for even $i$, a critical path with values $(d_1, d_2)$ for $d_1 \ge q_{(i / 2) \bmod{r}}$ and a critical path with values $(d_1', d_2')$ for $d_1' < q_{(i / 2) \bmod{r}}$ cannot share the same edge between layer $i$ and layer $i+1$. Similarly, for odd $i$, a critical path with values $(d_1, d_2)$ for $d_2 \ge q_{((i - 1) / 2) \bmod{r}}$ and a critical path with values $(d_1', d_2')$ for $d_2' < q_{((i - 1) / 2) \bmod{r}}$ cannot share the same edge between layer $i$ and layer $i+1$. 

    The rest of the proof follows essentially in the same way as the proof of \cref{lem:shortcut-properties} \cref{item:lem:shortcut-properties:item3}. 
    \end{itemize}
\end{proof}

\begin{proof}[Proof of \cref{thm:hopset}]
    The overall proof given \cref{lem:hopset-properties} is similar to the proof of \cref{thm:shortcut}, so we omit some details and highlight the differences. 

    We consider the graph $G=(V, E)$ and the critical paths $P$ constructed above. 
    
    Let $E_\Delta$ be the set of shortcut edges in $G$ that connects every pair of vertices $(u, v)$ where $u$ can reach $v$ using at most $\Delta$ edges, and the weight of each shortcut edge is the distance from $u$ to $v$. 

    We define a potential function with respect to any hopset $H$:
     \[
    \Phi(H) := \sum_{\pi \in P} \hopdist_{G \cup E_\Delta \cup H}(s_\pi, t_\pi). 
    \]
    Initially, we clearly have $\Phi(\emptyset) = |P| \cdot cr / \Delta = \Theta(c^3 r^6 / \Delta)$. Similar to the proof of \cref{thm:shortcut}, \[\Phi(H) - \Phi(H \cup \{e\}) = 
        O\left(\frac{cr}{\Delta} + \frac{r^2}{\Delta^2}\right).
        \]
    Therefore, if $|H| = \eps_1 \min\{c^2 r^5, c^3 r^4 \Delta\}$ for some sufficiently small constant $\eps_1 > 0$, $\Phi(\emptyset) - \Phi(H) \le \eps_2 c^3 r^6 / \Delta$ for some small constant $\eps_2$, which can be made arbitrarily small depending on $\eps_1$. As a result, $\Phi(H) = \Theta(c^3 r^6 / \Delta)$, which will imply that there exists a critical path $\pi$ where the minimum number of hops of any shortest path from $s_\pi$ to $t_\pi$ is $\Omega(\Phi(H) / |P|) = \Omega(cr / \Delta)$. 

    We set $c = r / B$ for some large constant $B$. By \cref{lem:hopset-properties} \cref{item:lem:hopset-properties:item1}, $|V|, |E| = \Theta(r^7 / B^3)$. The number of edges allowed in the hopset is $\eps_1 \min\{c^2 r^5, c^3 r^4 \Delta\} = \eps_1 \min\{c^7 / B^2, c^7 \Delta / B^3\}$, which can be any constant factor larger than $|E|$ by setting $B$ and $\Delta$ large enough. The hop bound is $\Omega(cr / \Delta) = \Omega(r^2 / B\Delta)$, which is $\Omega(|V|^{2/7})$ for constant values of $B$ and $\Delta$.     
\end{proof}

\section*{Acknowledgements}
We would like to thank Greg Bodwin for helpful discussions and feedback. 

\bibliographystyle{alpha}
\bibliography{ref}

\appendix

\section{Proof of Assumptions 3 and 4 from Section \ref{sec:reduction}}
\label{app:ass}
Our proof of Assumptions 3 and 4 will follow from an argument similar to the Cleaning Lemma of \cite{BHT22}. Note that the graph obtained after applying Lemma \ref{lem:cleaning} may be weighted. Consequently, this lemma is not useful for our unweighted directed distance preserver upper bound in Section \ref{sec:pres_up}.

\begin{lemma}
Let $G, P$ be an $n$-node directed, weighted graph and  a set of $|P| = p$ demand pairs, such that the minimal distance preserver $H$ of $G, P$ has $$|E(H)| = \Theta(\textnormal{\DDP}(n, p))$$ edges. Let $\pi(\cdot, \cdot)$ be a tiebreaking scheme associated with $G, P$. We can make the following assumptions on $G$, $P$, $H$,   and $\pi(\cdot, \cdot)$ without loss of generality:
\begin{enumerate}[start=3]
    \item Every node $v \in V(G)$ has degree at least  $\deg_H(v) \geq \frac{|E(H)|}{4n}$ in $H$.
    \item Every path $\pi \in \pi(P)$ has at least $|\pi| \geq \frac{|E(H)|}{4p}$ edges. 
\end{enumerate}
\label{lem:cleaning}
\end{lemma}
\begin{proof}
Firstly, we can assume that for all $(s, t) \in P$, $\pi(s, t)$ is a unique shortest path in $G$. Then, necessarily,
$$
H = \bigcup_{\pi \in \pi(P)} \pi.
$$
Secondly, we can assume without loss of generality that paths in $\pi(P)$ are pairwise edge-disjoint. Thirdly, we can assume that $G = H$, since deleting edges in $E(G) \setminus E(H)$ from $H$ does not affect the size of $H$. We will ensure Assumptions 3 and 4 hold without loss of generality by a simple procedure to modify $G, P$.  Let $\Pi = \pi(P)$ denote the collection of shortest paths in $G$ corresponding to demand pairs $P$. We modify $G, P,$ and $\Pi$ as follows. 
\begin{enumerate}
    \item Repeat the following two operations until no longer possible:
    \begin{enumerate}
        \item While there exists a path $\pi = \pi(s, t) \in \Pi$ such that $|\pi| < |E(G)|/(4p)$, delete path $\pi$ from $\Pi$. 
        \item While there exists a node $v \in V(G)$ such that $\deg_G(v) < |E(G)| / (4n)$, remove node $v$ from $G$. For each path $\pi \in \Pi$ that contains node $v$, we repair path $\pi$ as follows. If $v$ is the first node or the last node in path $\pi$, then we repair $\pi$ by deleting node $v$ from $\pi$. Otherwise, let $u$ and $w$ be nodes immediately preceding and succeeding $v$ in $\pi$, respectively. We delete node $v$ from $\pi$ and add an edge from $u$ to $w$ with weight $w((u, w)) = \dist_G(u, w)$. 
    \end{enumerate}
    \item We define our modified version of graph $G$ as 
    $$
G = \bigcup_{\pi \in \Pi} \pi.
    $$
    \item We define our modified version of demand pairs $P$ as
    $$
    P = \{(s, t) \in V(G) \times V(G) \mid \text{there exists an $s \leadsto t$ path in $\Pi$} \}.
    $$
\end{enumerate}
It is straightforward to verify that after our modifications, paths in $\Pi$ are unique shortest paths in $G$.  Moreover, Step 1(a) of our modification deletes at most $p \cdot |E(G)|/(4p) = |E(G)|/4$ edges from $H$, and Step 1(b) of our modification deletes at most $n \cdot |E(G)|/(4n)$ edges from $H$, so we conclude that after Step 1 terminates, there are at least $|E(G)|/2$ edges in $G$, so $G$ is nonempty. Then by our termination conditions, we must have that
\begin{itemize}
    \item Every path $\pi \in \pi(P)$ must satisfy $|\pi| \geq |E(G)|/(4p)$.
    \item Every node $v \in V(G)$ must satisfy $\deg_G(v) \geq |E(G)|/(4n)$. 
\end{itemize}
Finally, we observe that $|E(G)| = \Theta(\textsc{DDP}(n, p))$, since we deleted at most half of the edges originally in $G$. This establishes the lemma.
\end{proof}

\section{Proof of Lemma \ref{lem:sourcewise_pres} from Section \ref{sec:pres_up}}
The goal of this section is to prove Lemma \ref{lem:sourcewise_pres}, which we restate below for convenience. 

\begin{lemma}[cf. Lemma 8 of \cite{BV21}]
Let $G$ be an $n$-node directed graph, and let $S \subseteq V(G)$ be a set of $|S| = s$ nodes of weak diameter $h$ (i.e., $\dist_G(s, t) \leq h$ for all $s, t \in S$). Let $P \subseteq S \times V$ be a set of $|P| = p$ sourcewise demand pairs in $G$. Then there exists a distance preserver of $G, P$ of size
$$
O\left((nsph)^{1/2} + n \right).
$$
\end{lemma}
\label{app:swise}
The proof of this lemma follows from an argument similar to that of Lemma 8 of \cite{BV21}. We will require the following lemma due to \cite{CE06}.
\begin{lemma}[\cite{CE06}; cf. Lemma 7 of \cite{BV21}]
    If there are $b$ branching events associated with a tiebreaking scheme $\pi(\cdot, \cdot)$ of $G, P$, then the associated distance preserver $H$ has size $|E(H)| = O((nb)^{1/2}+n)$. 
    \label{lem:branching_event}
\end{lemma}

We will now prove Lemma \ref{lem:sourcewise_pres}.

\begin{proof}[Proof of Lemma \ref{lem:sourcewise_pres}.]
Let $\pi(\cdot, \cdot)$ be a tiebreaking scheme whose associated distance preserver 
$$
H = \bigcup_{(s, t) \in P} \pi(s, t)
$$
is minimal. For every edge $e \in E(H)$, we assign edge $e$ a label $\ell(e) \in P$. Specifically, we let $\ell(e) = (s, t) \in P$ for an arbitrary demand pair $(s, t) \in P$ such that every $s \leadsto t$ shortest path in $H$ uses edge $e$. We note that every edge $e \in E(H)$ can be labeled in this way -- otherwise, we could delete edge $e$ from $H$, contradicting the minimality of $H$.  Likewise, for each branching event $b = (e_1, e_2) \in E(H) \times E(H)$, we assign the branching event label $\ell(b) = (\ell(e_1), \ell(e_2))$. 

We will prove Lemma \ref{lem:sourcewise_pres} by bounding the number of branching events and applying Lemma \ref{lem:branching_event}. Suppose for the sake of contradiction that $H$ has greater than $sp(2h+1)$ branching events. Then by applying the pigeonhole principle, there is a node $t \in S$ and a demand pair $(u, v) \in P$ such that there are at least $2h+2$ distinct branching events $b_1, \dots, b_{2h+2}$ such that $\ell(b_i) = ((u, v), (t, w_i))$ for some $w_i \in V(H)$, for all $i \in [2h+2]$. For each branching event $b_i = (e_i^1, e_i^2)$, let $v_i \in V(H)$ be the unique node incident to $e_i^1$ and $e_i^2$. We may assume wlog that nodes $\{v_1, \dots, v_{2h+2}\}$ are ordered by nondecreasing distance $\dist_H(u, v_i)$. 
By the triangle inequality, we know that
$$
\dist_H(u, v_i) \leq \dist_H(u, t) + \dist_H(t, v_i).
$$
Subtracting $\dist_H(t, v_i)$ from both sides and using the fact that $S$ has diameter at most $h$, we obtain
$$
\dist_H(u, v_i) - \dist_H(t, v_i) \leq \dist_H(u, t) \leq h.
$$
Likewise, by the triangle inequality, 
$$
\dist_H(t, v_i) \leq \dist_H(t, u) + \dist_H(u, v_i),
$$
and using the fact that $S$ has diameter at most $h$, we obtain
$$
\dist_H(t, v_i) - \dist_H(u, v_i) \leq \dist_H(t, u) \leq h. 
$$
Combining our two bounds on $\dist_H(u, v_i) - \dist_H(t, v_i)$, we arrive at
$$
-h \leq \dist_H(u, v_i) - \dist_H(t, v_i) \leq h.
$$
Then by the pigeonhole principle, there exists branching events $b_i, b_j$ where $i \neq j \in [2h+2]$ such that 
$$
\dist_H(t, v_i) - \dist_H(u, v_i) = \dist_H(t, v_j) - \dist_H(u, v_j). 
$$
Rearranging, we obtain
$$
\dist_H(t, v_j) = \dist_H(t, v_i) + \dist_H(u, v_j) - \dist_H(u, v_i).
$$
Because we assumed that $\dist_H(u, v_i) \leq \dist_H(u, v_j)$  and $v_i, v_j \in \pi(u, v)$, it follows that
$$
\dist_H(u, v_j) - \dist_H(u, v_i) = \dist_H(v_i, v_j). 
$$
And therefore,
$$
\dist_H(t, v_j) = \dist_H(t, v_i) + (\dist_H(u, v_j) - \dist_H(u, v_i)) = \dist_H(t, v_i) + \dist_H(v_i, v_j). 
$$
This implies that the $t \leadsto v_j$ path
$$
\pi(t, v_i) \circ \pi(u, v)[v_i, v_j]
$$
is a shortest $t \leadsto v_j$ path in $H$, where $\pi(u, v)[v_i, v_j]$ denotes the $v_i \leadsto v_j$ subpath of path $\pi(u, v)$.
We conclude that there is a shortest $t \leadsto v_j$ path that uses an edge in $\pi(u, v)$ to reach node $v_j$. Consequently, not every shortest $t \leadsto v_j$ path includes an edge $e_j$  in the branching event $b_j$. This contradicts our labeling scheme of branching events.

We conclude that $H$ has $O(sph)$ branching events. Plugging in Lemma \ref{lem:branching_event}, we get that
$$
|E(H)| = O((nsph)^{1/2} + n),
$$
completing the proof of Lemma \ref{lem:sourcewise_pres}. 
\end{proof}

\end{document}